\documentclass[12pt,oneside]{amsart}

\usepackage{amsthm,amsmath,amssymb,amsfonts,amscd,mathtools,upgreek,url,
}
\usepackage{amsaddr}

\newcommand{\sho}{{\mathbf{S}}}
\newcommand{\tdt}{\mathbf{D}_t}

\newcommand{\pe}{u}
\newcommand{\vv}{w}
\newcommand{\tvv}{\tilde{\vv}}

\newcommand{\ik}{\ell}
\newcommand{\iup}{\xi}
\newcommand{\nt}{n}
\newcommand{\km}{k}
\newcommand{\id}{\mathbf{Id}}

\newcommand{\mX}{\mathbb{X}}
\newcommand{\mY}{\mathbb{Y}}
\newcommand{\mH}{\mathbb{H}}
\newcommand{\bC}{\mathbb{C}}
\newcommand{\kmx}{\upalpha}

\newcommand{\ial}{\upalpha}
\newcommand{\ibe}{\upbeta}
\newcommand{\imm}{q}
\newcommand{\glem}{\mathrm{G}}

\newcommand{\mm}{\mathbf{M}}

\newcommand{\gbf}{\mathbf{G}}

\newcommand{\gbb}{\mathbb{G}}

\newcommand{\dd}[1]{\partial_{#1}}
\newcommand{\ofc}[2]{of~class~$(#1,#2)$}

\newcommand{\ws}[1]{\mathbb{W}_{#1}}
\newcommand{\hs}[1]{\mathbb{B}_{#1}}

\newcommand{\euop}[1]{\boldsymbol\Delta^{#1}}

\newcommand{\mme}{\mathbf{m}}

\newcommand{\ap}{\mathrm{q}}

\newcommand{\mc}{\mathrm{c}}

\newcommand{\ud}{u_\ik,\ldots}
\newcommand{\ua}{{\mathrm{a}}}
\newcommand{\ub}{{\mathrm{b}}}
\newcommand{\tua}{{\tilde{\ua}}}
\newcommand{\tub}{{\tilde{\ub}}}

\newcommand{\fn}{f}
\newcommand{\mv}{\Tilde{u}}

\newcommand{\bv}{\Tilde{\Tilde{u}}}
\newcommand{\td}{w}

\newcommand{\mf}{\Phi}
\newcommand{\la}{\uplambda}
\newcommand{\nab}{\nabla}

\newcommand{\ff}{\mathbb{F}}
\newcommand{\tff}{\widetilde{\ff}}

\newcommand{\fik}{\mathbb{C}}

\newcommand{\sm}{\mathrm{d}}

\newcommand{\lb}{\label}
\newcommand{\er}{\eqref}

\newcommand{\zp}{\mathbb{Z}_{\ge 0}}
\newcommand{\zsp}{\mathbb{Z}_{>0}}
\newcommand{\zz}{\mathbb{Z}}

\newcommand{\cl}{\colon}
\newcommand{\pd}{\partial}
\newcommand{\lt}{\mathbf{U}}

\newcommand{\ndu}{N}

\newcommand{\ip}{p}

\newcommand{\tiltdt}{\widetilde{\mathbf{D}}_t}

\newcommand{\mer}{\mathcal{M\!E\!R}(\mathbb{C},\la)}

\newtheorem{theorem}{Theorem}

\newtheorem{lemma}{Lemma}

\theoremstyle{definition}
\newtheorem{definition}{Definition}

\newtheorem{remark}{Remark}

\begin{document}

\title[On matrix Lax representations for differential-difference equations]{On matrix Lax representations for (1+1)-dimensional \\ evolutionary differential-difference equations}
\date{}

\author{Sergei Igonin} 
\address{Center of Integrable Systems, P.G. Demidov Yaroslavl State University, Yaroslavl, Russia}
\email{s-igonin@yandex.ru}

\subjclass[2020]{37K60, 37K35}
\keywords{Integrable differential-difference equations; matrix Lax representations; Lax pairs; gauge transformations; Miura-type transformations}

\begin{abstract}
Differential-difference matrix Lax representations (Lax pairs), 
gauge transformations, and discrete Miura-type transformations (MTs)
belong to the main tools in the theory of (nonlinear) integrable differential-difference equations.

For a given equation, two matrix Lax representations (MLRs) 
are said to be gauge equivalent if one of them 
can be obtained from the other by applying a matrix gauge transformation.
Generalizing and extending several previous works on MLRs and MTs,
we present new results on the following problems:
\begin{itemize}
	\item When and how can one simplify a given MLR by means of gauge transformations?
	\item How can one use MLRs and gauge transformations for constructing MTs?
	\item A MLR is called fake if it is gauge equivalent to a trivial MLR.
How to determine whether a given MLR is not fake?
\end{itemize}

We consider the general (1+1)-dimensional evolutionary differential-difference case 
when a MLR can depend on any shifts of dependent variables and can be non-autonomous.
As applications and illustrations of the presented general theory, 
we construct several new two-component integrable equations (with new MLRs) 
connected by new MTs to known integrable equations 
from the papers [S.~Konstantinou-Rizos, A.V.~Mikhailov, P.~Xenitidis, \emph{J. Math. Phys.} 2015],
[E.~Mansfield, G.~Mar\'i Beffa, Jing~Ping~Wang, \emph{Found. Comput. Math.} 2013]),
including non-autonomous examples.
\end{abstract}

\maketitle

\section{Introduction}
\lb{sint}

In this paper we study relations between 
(1+1)-dimensional differential-difference matrix Lax representations (Lax pairs), 
gauge transformations, and (discrete) Miura-type transformations,
which belong to the main tools in the theory of (1+1)-dimensional 
integrable differential-difference equations.
Such equations occupy an important place in the modern theory of integrable systems 
and its applications.
In particular, such equations arise in discretizations of integrable 
partial differential equations (PDEs), 
in discretizations of various differential geometric constructions
and as chains associated with Darboux and B\"acklund transformations of PDEs 
(see, e.g.,~\cite{HJN-book2016,kmw,KMX2015,LWY-book2022,MMBW2013,suris2003} 
and references therein).

In this paper, $\zsp$ and $\zp$ denote the sets of positive and nonnegative integers, respectively.
Let $\ndu\in\zsp$. Below we consider an equation for an $\ndu$-component vector-function 
$u=\big(u^1(n,t),\ldots,u^{\ndu}(n,t)\big)$
of an integer variable~$n$ and a real or complex variable~$t$.
For any fixed integer~$\ik$, we have $u_\ik=(u^1_\ik,\ldots,u^{\ndu}_\ik)$,
where for each $\iup=1,\ldots,\ndu$ the component $u_\ik^\iup$ is a function
of $n,t$ defined as follows $u_\ik^\iup(n,t)=u^\iup(n\!+\!\ik,t)$.
That is, $u_\ik(n,t)=u(n\!+\!\ik,t)$. In particular, $u_0=u$.

Let $\ua,\ub\in\zz$, $\,\ua\le\ub$.
Consider a (1+1)-dimensional evolutionary differential-difference equation
\begin{gather}
\lb{sdde}
\pd_t(u)=\ff(\nt,u_\ua,u_{\ua+1},\ldots,u_\ub),
\end{gather}
where $\ff$ is an $\ndu$-component vector-function $\ff=(\ff^1,\ldots,\ff^{\ndu})$, 
which may depend on~$n$ and on~$u_\ik$ for $\ik=\ua,{\ua\!+\!1},\ldots,\ub$.


The differential-difference equation \eqref{sdde} 
is equivalent to the following infinite collection of differential equations
\begin{gather}
\lb{infcol}
\pd_t\big(u(n,t)\big)=\ff\big(\nt,u(n\!+\!\ua,t),u(n\!+\!\ua\!+\!1,t),\ldots,u(n\!+\!\ub,t)\big),\qquad\quad n\in\zz.
\end{gather}
Using the components of $u=\big(u^1(n,t),\ldots,u^{\ndu}(n,t)\big)$ 
and of $\ff=(\ff^1,\ldots,\ff^{\ndu})$, one can rewrite~\eqref{sdde} as
\begin{gather}
\lb{msdde}
\pd_t\big(u^i\big)=\ff^i(\nt,u_\ua,u_{\ua+1},\ldots,u_\ub),\qquad\quad
i=1,\ldots,\ndu.
\end{gather}
For each fixed $\ik\in\zz$, replacing $n$ by $n+\ik$ in~\er{infcol}, 
we see that \er{sdde} implies 
\begin{gather}
\lb{uildde}
\pd_t\big(u^i_\ik\big)=\ff^i(n\!+\!\ik,u_{\ua+\ik},u_{\ua+1+\ik},\ldots,u_{\ub+\ik}),
\qquad\quad
i=1,\ldots,\ndu,\qquad \ik\in\zz.
\end{gather}
We use the formal theory of (1+1)-dimensional evolutionary differential-difference equations, where one regards 
\begin{gather}
\lb{uldiu}
n,\qquad\qquad
u_\ik=(u^1_\ik,\ldots,u^\ndu_\ik),\quad\ik\in\zz,
\end{gather}
as independent quantities.
Note that $u_\ik=(u^1_\ik,\ldots,u^\ndu_\ik)$
are sometimes called \emph{dynamical variables}.

In this paper, the notation of the type $\fn=\fn(\nt,\ud)$ means 
that a function~$\fn$ depends on a finite number of the variables~\er{uldiu}.
The notation of the type $\fn=\fn(\nt,u_\ial,\ldots,u_\ibe)$ 
for some integers $\ial\le\ibe$ means that $\fn$ 
may depend on~$n$ and on~$u_\ik^{\xi}$ for $\ik=\ial,\ldots,\ibe$, 
$\,\xi=1,\ldots,\ndu$.


Equation~\er{sdde} said to be \emph{non-autonomous} 
if its right-hand side~$\ff$ depends nontrivially on~$n$.


We denote by~$\sho$ the \emph{shift operator} with respect to the variable~$n$.
For any function $g=g(n,t)$ one has the function~$\sho(g)$
such that $\sho(g)(n,t)=g(n\!+\!1,t)$.
Furthermore, for each $k\in\zz$, we have the $k$th power~$\sho^k$ 
of the operator~$\sho$ and the formula $\sho^k(g)(n,t)=g(n\!+\!k,t)$.

Since $u_\ik$ corresponds to $u(n+\ik,t)$, the operator~$\sho$ 
and its powers~$\sho^k$ for $k\in\zz$ act on functions of
the variables~\er{uldiu} by means of the rules
\begin{gather}
\lb{csuf}
\begin{gathered}
\sho(n)=n+1,\qquad\sho(u_\ik)=u_{\ik+1},\qquad
\sho^k(n)=n+k,\qquad\sho^k(u_\ik)=u_{\ik+k},\\
\sho^k\big(\fn(\nt,u_\ik,\ldots)\big)=\fn(n\!+\!k,\sho^k(u_{\ik}),\ldots).
\end{gathered}
\end{gather}
That is, applying $\sho^k$ to a function $\fn=\fn(\nt,u_\ik,\ldots)$, 
we replace~$n$ by~$n\!+\!k$ in~$\fn$ 
and replace~$u_\ik^{\xi}$ by~$u_{\ik+k}^{\xi}$ in~$\fn$ for all $\ik\in\zz$,
$\,\xi=1,\ldots,\ndu$.

The \emph{total derivative operator~$\tdt$ corresponding to equation~\eqref{sdde}} 
acts on functions of the variables~\er{uldiu} as follows
\begin{gather}
\lb{dtfu}
\tdt\big(\fn(\nt,\ud)\big)=\sum_{\ik,\xi}\sho^\ik(\ff^\xi)\cdot\frac{\pd \fn}{\pd u_\ik^\xi},
\end{gather}
where $\ff^\xi=\ff^\xi(\nt,u_\ua,\ldots,u_\ub)$, $\,\xi=1,\ldots,\ndu$, 
are the components of the vector-function 
$\ff=(\ff^1,\ldots,\ff^\ndu)$ from~\eqref{sdde}.
Formula~\er{dtfu} reflects the chain rule for the derivative with respect to~$t$, 
taking into account equations~\er{uildde}.
From~\er{csuf},~\er{dtfu} it follows that, for any function $h=h(\nt,\ud)$,
one has $\tdt\big(\sho(h)\big)=\sho\big(\tdt(h)\big)$.

In this paper, matrix-functions are sometimes called simply matrices.

\begin{remark}
\lb{rratf}
In this paper we deal with functions which may depend 
on a finite number of the variables~\er{uldiu} and a complex parameter~$\la$.
For simplicity of exposition, in order to present the main ideas clearly and concisely, 
we consider only functions which depend rationally on
the variables~\er{uldiu} and meromorphically on~$\la$, in the sense described below.
Many results of this paper can be extended to more general classes of functions, 
but such extension requires more technical and heavy proofs.

A scalar ($\bC$-valued) function $\fn=\fn(\nt,\ud,\la)$ 
is said to be \emph{of admissible type} if $\fn$ can be written in the form
\begin{gather}
\lb{fngh}
\fn=\frac{g(\nt,\ud,\la)}{h(\nt,\ud,\la)},
\end{gather}
where $g(\nt,\ud,\la)$ and $h(\nt,\ud,\la)$ are polynomials in the variables~\er{uldiu} 
whose coefficients are meromorphic functions of~$\la\in\bC$, 
and $h(\nt,\ud,\la)$ is not the zero polynomial.

We denote by~$\mer$ the field of meromorphic functions of~$\la\in\bC$.
Then~\er{fngh} can be regarded as a rational function of $\nt,\,u_\ik=(u^1_\ik,\ldots,u^{\ndu}_\ik)$
with coefficients in~$\mer$.
Thus, a scalar function $\fn=\fn(\nt,\ud,\la)$ is of admissible type if $\fn$
is a rational function of the variables~\er{uldiu} with coefficients in~$\mer$.

For each $\sm\in\zsp$, when we consider a $\sm\times\sm$ matrix-function
$\mm=\mm(\nt,\ud,\la)$, we assume that each entry of the matrix~$\mm$ 
is a scalar function of admissible type.
When we say that such a matrix~$\mm(\nt,\ud,\la)$ is assumed to be invertible, 
we mean that the determinant $\det\big(\mm(\nt,\ud,\la)\big)$ is not identically zero. 
Then $\big(\det\big(\mm(\nt,\ud,\la)\big)\big)^{-1}$ is well defined 
(as a rational function of $\nt,\,u_\ik=(u^1_\ik,\ldots,u^{\ndu}_\ik)$ with coefficients in~$\mer$), 
which allows us to compute the inverse matrix $\mm^{-1}=\big(\mm(\nt,\ud,\la)\big)^{-1}$, 
whose entries are of admissible type.
\end{remark}

\begin{definition}
\lb{dmlpgt}
Let $\sm\in\zsp$. 
Let $\mm=\mm(\nt,\ud,\la)$ and $\lt=\lt(\nt,\ud,\la)$ be $\sm\times\sm$ matrix-functions
depending on the variables~\er{uldiu} and a complex parameter~$\la$.
Suppose that $\mm$ is invertible (in the sense of Remark~\ref{rratf}) and we have
\begin{gather}
\lb{lr}
\tdt(\mm)=\sho(\lt)\cdot\mm-\mm\cdot\lt,
\end{gather}
where $\tdt$ is given by~\eqref{dtfu}.
Then the pair $(\mm,\lt)$ is called a \emph{matrix Lax representation} (MLR) for equation~\eqref{sdde}.

The connection between~\eqref{lr} and~\eqref{sdde} is as follows.
The components~$\ff^\iup$, $\,\iup=1,\ldots,\ndu$, of the right-hand side $\ff(\nt,u_\ua,\ldots,u_\ub)$ 
of~\eqref{sdde} appear in formula~\er{dtfu} for the operator~$\tdt$, which is used in~\er{lr}.

We say that a MLR $\big(\mm(\nt,\ud,\la),\,\lt(\nt,\ud,\la)\big)$ is \emph{non-autonomous} 
if at least one of the matrix-functions $\mm,\,\lt$ depends nontrivially on~$n$.

Correspondingly, a MLR $\big(\mm,\lt\big)$ is called \emph{autonomous} if $\mm,\,\lt$ do not depend on~$n$. 

Relation~\er{lr} implies that the following (overdetermined) auxiliary linear system
\begin{gather}
\lb{syspsi}
\sho(\Psi)=\mm\cdot\Psi,\qquad\quad
\pd_t(\Psi)=\lt\cdot\Psi
\end{gather}
is compatible modulo equation~\eqref{sdde}. 
Here $\Psi=\Psi(n,t)$ is an invertible $\sm\times\sm$ matrix-function.

We say that the matrix $\mm=\mm(\nt,\ud,\la)$ is the \emph{$\sho$-part} of the MLR $(\mm,\lt)$.

Then for any invertible $\sm\times\sm$ matrix-function  
$\gbf=\gbf(\nt,\ud,\la)$ the matrices
\begin{gather}
\lb{ggtmtu}
\widehat \mm=\sho(\gbf)\cdot \mm\cdot\gbf^{-1},\qquad\quad
\widehat\lt=\tdt(\gbf)\cdot\gbf^{-1}+
\gbf\cdot\lt\cdot\gbf^{-1}
\end{gather}
form a MLR for equation~\eqref{sdde} as well.
Here the matrix $\gbf=\gbf(\nt,\ud,\la)$
is invertible in the sense of Remark~\ref{rratf} and is called a \emph{gauge transformation}.
The MLR $\big(\widehat \mm,\,\widehat\lt\big)$ is \emph{gauge equivalent} 
to the MLR $(\mm,\lt)$, and one can say that
$\big(\widehat \mm,\,\widehat\lt\big)$ is obtained from $(\mm,\lt)$ 
by means of the gauge transformation~$\gbf$.
\end{definition}

\begin{definition}
\lb{dtrmlp}
We consider MLRs for a given equation~\er{sdde}.
A MLR~$(\mm,\lt)$ is said to be \emph{trivial} 
if $\mm$ does not depend on~$u_\ik$ for any $\ik\in\zz$.
Then \er{dtfu} yields $\tdt(\mm)=0$, and from~\er{lr} one derives 
$\sho(\lt)=\mm\cdot\lt\cdot\mm^{-1}$, which 
implies that $\lt$ does not depend on~$u_\ik$ either.
(One can show this easily, using the rules~\er{csuf}).
Hence a trivial MLR does not provide any information about equation~\er{sdde}.

A MLR is called \emph{fake} if it is gauge equivalent to a trivial MLR.
That is, fake MLRs are obtained from trivial MLRs by applying gauge transformations, 
and this can be done independently of equation~\er{sdde}.
\end{definition}

\begin{definition}
\lb{dwd}
Fix $\sm\in\zsp$.
Let $\ws{\sm}$ be the vector space of all $\sm\times\sm$ matrix-functions of the form $Q=Q(\nt,\ud,\la)$.
Let $\hs{\sm}\subset\ws{\sm}$ be the subset of all invertible $\sm\times\sm$ matrix-functions $M=M(\nt,\ud,\la)$.
For each $M\in\hs{\sm}$, we consider the operator $\nab_{M}\cl\ws{\sm}\to\ws{\sm}$ given by the formula
\begin{gather}
\lb{nbmq}
\nab_{M}(Q)=\sho(M^{-1}QM).
\end{gather}
The operator $\nab_{M}$ is invertible, and one has
\begin{gather}
\lb{nbm1q}
\nab_{M}^{-1}(Q)=M\cdot\sho^{-1}(Q)\cdot M^{-1}.
\end{gather}
For each $k\in\zz$, we denote by $\nab_{M}^k$ the $k$th power 
of the operator~$\nab_{M}$. In particular, one has
\begin{gather*}
\nab_{M}^0(Q)=Q,\qquad\quad
\nab_{M}^2(Q)=\nab_{M}(\nab_{M}(Q))=\sho(M^{-1}\sho(M^{-1}QM)M),\\
\nab_{M}^{-2}(Q)=\nab_{M}^{-1}(\nab_{M}^{-1}(Q))=M\sho^{-1}(M\sho^{-1}(Q)M^{-1})M^{-1}.
\end{gather*}
Furthermore, for each $j=1,\ldots,\ndu$, we consider the map $\euop{u^j}\cl\hs{\sm}\to\ws{\sm}$ given by
\begin{gather}
\lb{foreo}
\euop{u^j}(M)=\sum_{\ik\in\zz}\nab_M^{-\ik}\big(\dd{u^j_\ik}(M)\cdot M^{-1}\big),\qquad\quad
M=M(\nt,\ud,\la)\in\hs{\sm},\qquad
j=1,\ldots,\ndu.
\end{gather}
For instance, 
if $M=M\big(\nt,u^\xi_{-1},u^\xi_{0},u^\xi_{1},\la\big)$ depends only on~$\nt$, $u^\xi_{-1}$, 
$u^\xi_{0}$, $u^\xi_{1}$, $\,\xi=1,\ldots,\ndu$, and on~$\la$, then 
\begin{multline*}
\euop{u^j}\big(M\big(\nt,u^\xi_{-1},u^\xi_{0},u^\xi_{1},\la\big)\big)=
\sum_{\ik\in\zz}\nab_M^{-\ik}\big(\dd{u^j_\ik}(M)\cdot M^{-1}\big)=
\sum_{\ik=-1,0,1}\nab_M^{-\ik}\big(\dd{u^j_\ik}(M)\cdot M^{-1}\big)=\\
=
\sho\big(M^{-1}\big(\dd{u^j_{-1}}(M)\cdot M^{-1}\big)M\big)
+\dd{u^j_0}(M)\cdot M^{-1}+M\sho^{-1}\big(\dd{u^j_1}(M)\cdot M^{-1}\big)M^{-1}.
\end{multline*}
\end{definition}

\begin{definition}
\lb{dord}
Let $\ial,\ibe\in\zz$, $\,\ial\le\ibe$. A MLR $\big(\mm(n,\ud,\la),\,\lt(n,\ud,\la)\big)$ 
is said to be \emph{\ofc{\ial}{\ibe}} 
if its \mbox{$\sho$-part}~$\mm(n,\ud,\la)$ is of the form $\mm=\mm(n,u_\ial,\dots,u_\ibe,\la)$. 
That is, a MLR $(\mm,\lt)$ is \ofc{\ial}{\ibe} if $\mm$ may depend only on $n,u_\ial,\dots,u_\ibe,\la$,
while $\lt$ may depend on any finite number of the variables~\er{uldiu} and on~$\la$.
\end{definition}

Let $\ial,\ibe,\imm\in\zz$, $\,\ial\le\ibe$.
As shown in Lemma~\ref{lgs} in Section~\ref{srmlp}, 
any MLR $(\mm,\lt)$ \ofc{\ial}{\ibe} is gauge equivalent 
to some MLR \ofc{\ial\!+\!\imm}{\ibe\!+\!\imm} 
with $\sho$-part equal to~$\sho^\imm(\mm)$.
Therefore, since we want to study MLRs up to gauge equivalence,
it is sufficient to consider MLRs~\ofc{0}{k} for $k\in\zp$.

In Section~\ref{srmlp} we present the following results.
\begin{itemize}
\item Let $\ip\in\zsp$. 
For a given MLR~\ofc{0}{\ip} with $\sho$-part $\mm(\nt,u_0,u_1,\ldots,u_\ip,\la)$,
Theorem~\ref{thmlpk1} provides
a necessary and sufficient condition for the possibility to 
simplify the MLR by means of gauge transformations in the following sense: \\
the possibility to eliminate the dependence on~$u_\ip$ in the $\sho$-part 
by applying a gauge transformation (i.e.~to obtain a gauge equivalent MLR 
with $\sho$-part of the form $\widehat{\mm}(\nt,u_0,\ldots,u_{\ip-1},\la)$).
If a MLR satisfies this condition (given in~\er{pdm}), 
then Theorem~\ref{thmlpk1} provides an explicit formula 
for a suitable gauge transformation (see~\er{guum}).

In other words, if we have a MLR $(\mm,\lt)$ \ofc{0}{\ip}, then
Theorem~\ref{thmlpk1} gives a necessary and sufficient condition for
existence of a gauge transformation $\gbf=\gbf(\nt,\ud,\la)$ 
such that the corresponding gauge equivalent MLR~\er{ggtmtu} is \ofc{0}{\ip\!-\!1}.
Applications of this theorem are discussed below.

(Note that Theorem~\ref{thmlpk1} in Section~\ref{srmlp} 
generalizes and extends~\cite[Theorem~1]{IgSimpl2024},
since \cite[Theorem~1]{IgSimpl2024} considers only 
autonomous MLRs and gives only a sufficient condition for the described problem 
to simplify a MLR by means of gauge transformations.)

\item 
For each $j=1,\ldots,\ndu$, using the map $\euop{u^j}\cl\hs{\sm}\to\ws{\sm}$ defined by~\er{foreo},
in Theorem~\ref{th1} we obtain property~\er{gmg}, which essentially says the following: \\
when an invertible matrix $M=M(\nt,\ud,\la)$ is transformed to 
$\widetilde{M}=\sho(G)\cdot M\cdot G^{-1}$ by a gauge transformation $G=G(\nt,\ud,\la)$, 
the corresponding matrix $\euop{u^j}\big(\widetilde{M}\big)$ can be obtained 
from $\euop{u^j}(M)$ by conjugation as follows 
$\euop{u^j}\big(\widetilde{M}\big)=\sho(G)\cdot\euop{u^j}(M)\cdot\sho(G)^{-1}$.

Note that in Remark~\ref{rinvg} 
we discuss relations of our results on differential-difference MLRs
with the results of  S.Yu.~Sakovich~\cite{sakov95}
and M.~Marvan~\cite{marvan93,marvan97} 
on zero-curvature representations of partial differential equations.

Using Theorems~\ref{thmlpk1},~\ref{th1} and Lemmas~\ref{lgs}--\ref{leop0}, 
we prove Theorem~\ref{thtriv} described below.
Further applications of Theorem~\ref{th1} are discussed in Remark~\ref{rinvg}.

\item Theorem~\ref{thtriv} says that a MLR $\big(\mm(n,\ud,\la),\,\lt(n,\ud,\la)\big)$
is gauge equivalent to a trivial MLR if and only if 
one has $\euop{u^j}(\mm)=0$ for all $j=1,\ldots,\ndu$.
Therefore, a MLR $(\mm,\lt)$ is not fake (in the sense of Definition~\ref{dtrmlp})
if and only if there is $j\in\{1,\ldots,\ndu\}$ such that $\euop{u^j}(\mm)\neq 0$.

Note that the question ``how to determine whether a given MLR is not fake'' 
was studied in~\cite[Theorem~2]{IgJGP2025} for a much smaller class of MLRs 
(namely, for MLRs with $\sho$-part of the form $\mm(u_0,u_1,u_2,\la)$), using a different method.
\end{itemize}

\begin{remark}
\lb{rinvg}
Using Theorem~\ref{th1} and other properties of the maps $\nab_{M}$, $\,\euop{u^j}$ 
discussed in Definition~\ref{dwd} and in Section~\ref{srmlp},
in the preprint~\cite{gaprep2026} we introduce and study invariants 
with respect to the action of gauge transformations on MLRs for 
(1+1)-dimensional evolutionary differential-difference (semi-discrete) equations.
As explained in~\cite{gaprep2026}, 
our approach to gauge invariants for differential-difference (semi-discrete) MLRs
is inspired by the results of S.Yu.~Sakovich~\cite{sakov95} and
M.~Marvan~\cite{marvan93,marvan97} on 
gauge transformations and gauge invariants 
for zero-curvature representations (ZCRs) of partial differential equations.

Also, as discussed in~\cite{gaprep2026},
our Definition~\ref{dwd} and Theorems~\ref{th1},~\ref{thtriv} in the differential-difference setting
are inspired by the approaches of S.Yu.~Sakovich~\cite{sakov95}
and M.~Marvan~\cite{marvan93,marvan97} to ZCRs of partial differential equations.
\end{remark}

\begin{remark}
\lb{rrelab}
According to Definition~\ref{dmlpgt}, 
in a MLR~$(\mm,\lt)$ for equation~\er{sdde}
the matrix~$\mm=\mm(\nt,\ud,\la)$ may depend on~$n$, 
on a finite number of the 
dynamical variables~$u_{\ik}=(u^1_{\ik},\ldots,u^\ndu_{\ik})$, $\,\ik\in\zz$, 
and a parameter~$\la$. For any fixed integers $q_1,\ldots,q_\ndu$, we can relabel 
\begin{gather}
\lb{relab}
u^1_\ik\,\mapsto\,u^1_{\ik+q_1},\,\quad\ldots\,\quad,\,\quad 
u^{\ndu}_\ik\,\mapsto\,u^{\ndu}_{\ik+q_\ndu}\qquad\quad\forall\,\ik\in\zz.
\end{gather}
Relabeling~\er{relab} means that in equation~\er{sdde} 
we make the following invertible change of variables 
\begin{gather}
\notag
u^1(n,t)\,\mapsto\,u^1(n\!+\!q_1,t),\,\quad\ldots\,\quad,\,\quad 
u^{\ndu}(n,t)\,\mapsto\,u^{\ndu}(n\!+\!q_\ndu,t).
\end{gather}
Applying a relabeling~\er{relab} with suitable fixed $q_1,\ldots,q_\ndu\in\zz$, 
one can transform the matrix~$\mm$ to the form 
$\mm=\mm(\nt,u_0,\ldots,u_k,\la)$ for some $k\in\zp$
(so that the MLR becomes~\ofc{0}{k}).
\end{remark}

In what follows, for any function $w=w(n,t)$ and each $\ell\in\zz$, 
we denote by~$w_\ell$ the function $w_\ell(n,t)=w(n+\ell,t)$.
That is, $w_\ell=\sho^\ell(w)$. In particular, $w_0=w$.

Now let $\tua,\tub\in\zz$, $\,\tua\le\tub$, 
and consider another (1+1)-dimensional evolutionary differential-difference equation
\begin{gather}
\lb{vdde}
\pd_t(\mv)=\tff(\nt,\mv_\tua,\mv_{\tua+1},\dots,\mv_\tub)
\end{gather}
for an $\ndu$-component vector-function $\mv=\big(\mv^1(n,t),\dots,\mv^{\ndu}(n,t)\big)$.
In the right-hand side of~\er{vdde} one has
an $\ndu$-component vector-function $\tff=\big(\tff^1,\dots,\tff^{\ndu}\big)$.

\begin{definition}
\lb{defmtt}
A \emph{Miura-type transformation} (MT) from equation~\eqref{vdde} 
to equation~\eqref{sdde} is determined by an expression of the form
\begin{gather}
\lb{uvf}
u=\mf(\nt,\mv_\ik,\dots)
\end{gather}
with the following requirements:
\begin{enumerate}
	\item The right-hand side of~\er{uvf} 
	is an $\ndu$-component vector-function $\mf=(\mf^1,\dots,\mf^{\ndu})$ 
	which  may depend on~$n$ and on a finite number of the dynamical variables 
	$\mv_\ik=\big(\mv^1_\ik,\dots,\mv^\ndu_\ik\big)$, $\,\ik\in\zz$.
	\item If $\mv=\mv(n,t)$ satisfies equation~\eqref{vdde} 
	then $u=u(n,t)$ determined by~\eqref{uvf} satisfies equation~\eqref{sdde}.
\end{enumerate}
More precisely, the second requirement means that we must have relations~\er{dtmfi} explained below.
The components of the vector formula~\eqref{uvf} are 
\begin{gather}
\lb{uimfi}
u^i=\mf^i\big(\nt,\mv^\iup_\ik,\dots\big),\qquad\quad i=1,\dots,\ndu.
\end{gather}
Recall that \er{sdde} is equivalent to~\er{msdde}.
Substituting the right-hand side of~\eqref{uimfi} in place of~$u^i$ in~\eqref{msdde}, we get 
\begin{gather}
\lb{dtmfi}
\tiltdt\big(\mf^i\big(\nt,\mv^\iup_\ik,\dots\big)\big)=
\ff^i\big(\nt,\sho^\ua(\mf),\sho^{\ua+1}(\mf),\dots,\sho^\ub(\mf)\big),\qquad\quad
i=1,\dots,\ndu,
\end{gather}
which must be valid identically in the variables $n$ and $\mv^\xi_\ell$.
Here $\tiltdt$ is the total derivative operator corresponding to~\eqref{vdde}, 
so that in the left-hand side of~\er{dtmfi} we have 
$\tiltdt\big(\mf^i\big(\nt,\mv^\iup_\ik,\dots\big)\big)=%
\sum_{\ell,\xi}\sho^\ell\big(\tff^\xi\big)\cdot\big({\pd\mf^i}/{\pd\mv^\xi_\ell}\big)$,
where $\tff^\xi=\tff^\xi\big(\nt,\mv_\tua,\dots,\mv_\tub\big)$,
$\,\xi=1,\ldots,\ndu$, are the components of the vector-function 
$\tff=\big(\tff^1,\dots,\tff^{\ndu}\big)$ from~\eqref{vdde}.
\end{definition}
MTs for differential-difference equations
are a discrete analog of MTs for partial differential equations~\cite{miura1968}.
MTs for partial differential equations are sometimes called \emph{differential substitutions}.

\begin{definition}
\lb{dtame}
We consider MLRs for a given equation~\er{sdde} imposed 
on an $\ndu$-component vector-function $u=(u^1,\ldots,u^{\ndu})$.
We say that a MLR $(\mm,\lt)$ is \emph{$u$-tame}
if its $\sho$-part is of the form $\mm=\mm(\nt,u_0,\la)$.
That is, $(\mm,\lt)$ is $u$-tame if $\mm$ may depend only on $\nt$, 
$u_0=(u^1_0,\ldots,u^{\ndu}_0)$, $\la$.
\end{definition}

\begin{remark}
\lb{rgnla}
Consider a $u$-tame MLR $\big(\mm(\nt,u_0,\la),\,\lt(\nt,\ud,\la)\big)$.
Let $G=G(\nt,\la)$ be a gauge transformation such that $G$ may depend only on $\nt,\,\la$.
Then the corresponding gauge equivalent MLR
\begin{gather}
\notag
\check\mm=\sho(G)\cdot\mm\cdot G^{-1},\qquad\quad
\check\lt=\tdt(G)\cdot G^{-1}+G\cdot\lt\cdot G^{-1}
\end{gather}
is $u$-tame as well, since its $\sho$-part is
$\check\mm=\sho(G)\cdot\mm\cdot G^{-1}=G(\nt\!+\!1,\la)\cdot\mm(\nt,u_0,\la)\cdot G(\nt,\la)^{-1}$.
\end{remark}

The following observation is well known: \\
when one makes a classification for some class 
of integrable (partial differential, difference or differential-difference) 
equations with two independent variables,
it is often possible to select a few basic equations (with some parameters)
such that all the other equations from the considered class can be derived from 
the basic ones by means of MTs 
(see, e.g.,~\cite{drin-sok85,mss91,yam2006,garif2017,garif2018,meshk2008,LWY-book2022} 
and references therein).
Also, it is well known that MTs often help to obtain conservation laws 
and B\"acklund transformations for equations of various types
(see, e.g.,~\cite{miura-cl,suris2003,HJN-book2016,LWY-book2022}).
Hence it is desirable to develop constructions of MTs.

In Section~\ref{smt} we demonstrate a technique to construct MTs 
(including non-autonomous examples)
by means of MLRs and gauge transformations, using the following steps:
\begin{itemize}
	\item Having a MLR, which is not necessarily $u$-tame, we first try to bring it 
	to a $u$-tame form, using a relabeling from Remark~\ref{rrelab} and 
	gauge transformations~$\gbf(\nt,\ud,\la)$ from Theorem~\ref{thmlpk1}.
	\item Having obtained a $u$-tame MLR 
	$\big(\mm(\nt,u_0,\la),\,\lt(\nt,\ud,\la)\big)$, we try to simplify it by means of 
	gauge transformations of the form~$G(\nt,\la)$ depending only on $\nt,\,\la$
	in agreement with Remark~\ref{rgnla}.
	\item After that, one can often derive MTs 
	by means of reductions of the auxiliary system~\er{syspsi} 
corresponding to the obtained simplified MLR.
(Note that, in some examples in Section~\ref{smt}, we use~$\mv$ instead of~$u$.)
\end{itemize}
This technique extends the methods of~\cite{IgSimpl2024,IgJGP2025,ChisIg2025}, 
since in~\cite{IgSimpl2024,IgJGP2025,ChisIg2025} only the autonomous case was considered.

As applications and illustrations of the presented general theory on MLRs, 
we construct several new integrable equations (with new MLRs) connected 
by new MTs to known integrable equations from~\cite{KMX2015,MMBW2013}.
Namely, in Section~\ref{smt} we obtain the following results.
\begin{itemize}
	\item We consider the integrable 
	$2$-component (non-autonomous differential-difference) equation~\er{kmx} with 
	the MLR~\er{mlrbc}, which were derived by 
	S.~Konstantinou-Rizos, A.V.~Mikhailov, P.~Xenitidis~\cite{KMX2015}  
	in a study of Darboux transformations for a Lax operator 
	of the (partial differential) nonlinear Schr\"odinger equation. 
	In~\er{kmx},~\er{mlrbc} one has an arbitrary function~$\kmx(n)$.
	(Equation~\er{kmx} is given in~\cite[page~7]{KMX2015} in different notation,
using functions $p(n,x)$, $q(n,x)$, instead of $\pe^1(n,t)$, $\pe^2(n,t)$.)

We construct new integrable $2$-component 
(non-autonomous differential-difference) equations \er{dtmv}, \er{tdtmv}
and new MTs~\er{mtpeqq},~\er{tmtkmx}. 
The MT~\er{mtpeqq} is from~\er{dtmv} to~\er{kmx}, 
while the MT~\er{tmtkmx} is from~\er{tdtmv} to~\er{dtmv}.
This implies that the composition of the MTs~\er{mtpeqq},~\er{tmtkmx} 
gives a MT from~\er{tdtmv} to~\er{kmx}.

In formulas \er{dtmv}, \er{tdtmv}, \er{mtpeqq}, \er{tmtkmx} we have arbitrary constants~$\mc,\ap\in\fik$.
In~\er{tmtkmx},~\er{tdtmv} we have also functions~$\mX$,~$\mY$ given by~\er{xy}.
Furthermore, in~\er{tdtmv} one has the function~$\mH$ obtained by substituting~\er{tuellmt} 
in the right-hand side of the second component of~\er{dtmv}.
(We do not present the explicit formula for~$\mH$, since it is rather cumbersome.)

\item We consider also the integrable equation~\er{MMBW} derived by
E.~Mansfield, G.~Mar\'i Beffa, Jing~Ping~Wang \cite[Section~5.1.1]{MMBW2013}
in a study of invariant time evolutions of polygons in the centro-affine plane.
As shown in~\cite[Section~5.1.1]{MMBW2013}, one has the MT~\er{mtmbw}
from~\er{MMBW} to equation~\er{todafm}, 
which is the Toda lattice in the Flaschka--Manakov coordinates~\cite{FlaschPhR74,Manak75}.

We construct a new integrable $2$-component equation~\er{ceqnew} 
and a new MT~\er{mtmtd} from~\er{ceqnew} to~\er{MMBW}.
The composition of the MTs~\er{mtmbw},~\er{mtmtd} gives a MT from~\er{ceqnew} 
to the Toda lattice~\er{todafm}.
\item Furthermore, we obtain several new MLRs: \\
the MLR~\er{2hhm} for equation~\er{dtmv}, the MLR~\er{hmmtd} for equation~\er{MMBW},
and the MLRs described in Remarks~\ref{rlax1},~\ref{rlax2} for equations~\er{tdtmv},~\er{ceqnew}.
Each of these MLRs has a parameter~$\la\in\fik$.
\end{itemize}

The constructed equations~\er{dtmv},~\er{tdtmv},~\er{ceqnew} are integrable
for the following reasons:
\begin{itemize}
	\item As explained above, equations~\er{dtmv},~\er{tdtmv} are connected by MTs 
	to the integrable equation~\er{kmx} from~\cite{KMX2015},
	while \er{ceqnew} is connected by MTs 
	to the integrable equation~\er{MMBW} from~\cite[Section~5.1.1]{MMBW2013}
	and to the Toda lattice~\er{todafm}.
	\item As discussed above, each of equations~\er{dtmv},~\er{tdtmv},~\er{ceqnew}
	possesses a MLR with parameter~$\la\in\fik$.
\end{itemize}

\section{General results on differential-difference matrix Lax representations}
\lb{srmlp}

Let $\sm\in\zsp$.
In this section, all matrix-functions (including gauge transformations)
are of size~$\sm\times\sm$.

\begin{lemma}
\label{lgs}
Let $\ial,\ibe\in\zz$, $\,\ial\le\ibe$.
Let $\big(\mm(n,u_\ial,\dots,u_\ibe,\la),\,\lt(n,\ud,\la)\big)$ be a MLR \ofc{\ial}{\ibe}.
Since the matrix $\mm=\mm(n,u_\ial,\dots,u_\ibe,\la)$ is invertible, 
we can consider the gauge transformations
\begin{gather}
\lb{gtm}
\glem:=\mm,\qquad\quad\breve{\glem}:=\sho^{-1}(\mm^{-1})
\end{gather}
and the corresponding gauge equivalent MLRs
\begin{gather}
\lb{gmlr}
\widehat{\mm}:=\sho(\glem)\cdot \mm\cdot\glem^{-1},\qquad\quad
\widehat\lt:=\tdt(\glem)\cdot\glem^{-1}+\glem\cdot\lt\cdot\glem^{-1},\\
\lb{bgmlr}
\breve{\mm}:=\sho\big(\breve{\glem}\big)\cdot \mm\cdot\breve{\glem}^{-1},\qquad\quad
\breve{\lt}:=\tdt\big(\breve{\glem}\big)\cdot\breve{\glem}^{-1}+\breve{\glem}\cdot\lt\cdot\breve{\glem}^{-1}.
\end{gather}
Then we have
\begin{gather}
\lb{spts}
\widehat{\mm}=\sho(\mm),\qquad\quad
\breve{\mm}=\sho^{-1}(\mm),
\end{gather}
which implies that the MLR~\er{gmlr} is \ofc{\ial\!+\!1}{\ibe\!+\!1}, and
the MLR~\er{bgmlr} is \ofc{\ial\!-\!1}{\ibe\!-\!1}.

Furthermore, for each $\imm\in\zz$, the considered MLR $(\mm,\lt)$ \ofc{\ial}{\ibe}
is gauge equivalent to some MLR \ofc{\ial\!+\!\imm}{\ibe\!+\!\imm} 
with $\sho$-part equal to~$\sho^\imm(\mm)$.
\end{lemma}
\begin{proof}
Substituting~\er{gtm} in~\er{gmlr},~\er{bgmlr}, we get~\er{spts}.
Therefore, the MLR $(\mm,\lt)$ \ofc{\ial}{\ibe} is gauge equivalent 
to the MLR~\er{gmlr} \ofc{\ial\!+\!1}{\ibe\!+\!1} with $\sho$-part equal to~$\sho(\mm)$
and is gauge equivalent 	to the MLR~\er{bgmlr} \ofc{\ial\!-\!1}{\ibe\!-\!1} 
with $\sho$-part equal to~$\sho^{-1}(\mm)$.

Now, the last statement of the lemma can be proved easily by induction on~$\imm$.
\end{proof}

\begin{remark}
\lb{rchv}
In Theorem~\ref{thmlpk1} below, 
we consider an invertible $\sm\times\sm$ matrix-function 
$$
\mm=\mm(\nt,u_0,u_1,\ldots,u_\ip,\la)
$$ 
with $\ip\in\zsp$, and we need to choose a constant vector $a_0\in\fik^\ndu$ so that 
the matrix $\mm(\nt,a_0,u_1,\ldots,u_\ip,\la)$ is well defined and is invertible.
Here $\mm(\nt,a_0,u_1,\ldots,u_\ip,\la)$ is obtained from~$\mm(\nt,u_0,u_1,\ldots,u_\ip,\la)$ 
by substituting $u_0=a_0$.
The term ``invertible'' is understood in the sense of Remark~\ref{rratf}.

For a given invertible matrix $\mm(\nt,u_0,\ldots,u_\ip,\la)$, 
such a vector $a_0\in\fik^\ndu$ can be chosen as follows. 
For $i,j\in\{1,\ldots,\sm\}$ the $(i,j)$ entry of the matrix $\mm$ 
is denoted by $\mme^{ij}=\mme^{ij}(\nt,u_0,\ldots,u_\ip,\la)$.
Since, according to Remark~\ref{rratf}, $\mme^{ij}(\nt,u_0,\ldots,u_\ip,\la)$ 
is assumed to be a rational function of~$\nt,\,u_\ik=(u^1_\ik,\ldots,u^{\ndu}_\ik)$ 
with coefficients in~$\mer$, where $\mer$ is the field of meromorphic functions of~$\la\in\bC$, we have
\begin{gather*}
\mme^{ij}(\nt,u_0,\ldots,u_\ip,\la)=\frac{f^{ij}(\nt,u_0,\ldots,u_\ip,\la)}{g^{ij}(\nt,u_0,\ldots,u_\ip,\la)},
\quad\qquad i,j=1,\ldots,\sm,\\
\det\big(\mm(\nt,u_0,\ldots,u_\ip,\la)\big)=\frac{P(\nt,u_0,\ldots,u_\ip,\la)}{Q(\nt,u_0,\ldots,u_\ip,\la)},
\end{gather*}
where 
\begin{itemize}
	\item $f^{ij}(\nt,u_0,\ldots,u_\ip,\la)$, $\,g^{ij}(\nt,u_0,\ldots,u_\ip,\la)$, 
	$\,P(\nt,u_0,\ldots,u_\ip,\la)$, $\,Q(\nt,u_0,\ldots,u_\ip,\la)$ 
	are some polynomials in the variables~$\nt,\,u_\ik=(u^1_\ik,\ldots,u^{\ndu}_\ik)$ with coefficients in~$\mer$,
	\item each of the polynomials
\begin{gather*}
g^{ij}(\nt,u_0,u_1,\ldots,u_\ip,\la),\qquad i,j=1,\ldots,\sm,\qquad
P(\nt,u_0,u_1,\ldots,u_\ip,\la),\qquad Q(\nt,u_0,u_1,\ldots,u_\ip,\la)
\end{gather*}
is not the zero polynomial.
\end{itemize}
Then there is an open dense subset $W\subset\fik^\ndu$ such that 
for any constant vector $a_0\in W$ each of the polynomials
\begin{gather*}
g^{ij}(\nt,a_0,u_1,\ldots,u_\ip,\la),\qquad i,j=1,\ldots,\sm,\qquad
P(\nt,a_0,u_1,\ldots,u_\ip,\la),\qquad Q(\nt,a_0,u_1,\ldots,u_\ip,\la)
\end{gather*}
is not the zero polynomial (as a polynomial in~$\nt,\,u_\ik$).
This implies that for any $a_0\in W$
the matrix $\mm(\nt,a_0,u_1,\ldots,u_\ip,\la)$ is well defined and is invertible 
(in the sense of Remark~\ref{rratf}).
\end{remark}

\begin{theorem}
\label{thmlpk1}
Let $\ip\in\zsp$. 
Consider an invertible $\sm\times\sm$ matrix-function $\mm=\mm(\nt,u_0,u_1,\ldots,u_\ip,\la)$. 
The following two properties are equivalent.
\begin{itemize}
\item \textbf{\textup{Property~1}}.
There is a gauge transformation~$\gbf=\gbf(\nt,\ud,\la)$ such that the matrix
\begin{gather}
\lb{hmgm}
\widehat\mm=\sho(\gbf)\cdot\mm(\nt,u_0,\ldots,u_\ip,\la)\cdot\gbf^{-1}
\end{gather}
is of the form 
\begin{gather}
\lb{hmhm}
\widehat{\mm}=\widehat{\mm}(\nt,u_0,\ldots,u_{\ip-1},\la).
\end{gather}
\item \textbf{\textup{Property~2}}.
The matrix $\mm=\mm(\nt,u_0,u_1,\ldots,u_\ip,\la)$ obeys
\begin{gather}
\lb{pdm}
\forall\, i,j=1,\ldots,\ndu\qquad\quad
\dd{u^i_0}\big(\dd{u^j_\ip}(\mm)\cdot\mm^{-1}\big)=0.
\end{gather}
\end{itemize}

Now suppose that $\mm$ satisfies~\er{pdm}. 
Fix a constant vector $a_0=(a^1_0,\ldots,a^\ndu_0)\in\fik^\ndu$ 
and consider the gauge transformation
\begin{gather}
\lb{guum}
\gbf(\nt,u_0,\ldots,u_{\ip-1},\la):=\sho^{-1}\big(\mm(\nt,a_0,u_1,\ldots,u_\ip,\la)^{-1}\big),
\end{gather}
where $\mm(\nt,a_0,u_1,\ldots,u_\ip,\la)$ is obtained 
by substituting $u_0=a_0$ in~$\mm(\nt,u_0,u_1,\ldots,u_\ip,\la)$.
Here a vector $a_0\in\fik^\ndu$ is chosen so that the matrix\/ 
$\mm(\nt,a_0,u_1,\ldots,u_\ip,\la)$ is well defined and is invertible.
\textup{(}Remark~\textup{\ref{rchv}} explains how to choose such a vector.\textup{)}

Then the matrix $\widehat{\mm}$ determined by~\er{hmgm}
is of the form $\widehat{\mm}=\widehat{\mm}(\nt,u_0,\ldots,u_{\ip-1},\la)$, and we have
\begin{multline}
\lb{hmmm}
\widehat\mm=\sho(\gbf)\cdot\mm(\nt,u_0,u_1,\ldots,u_\ip,\la)\cdot\gbf^{-1}=\\
=\mm(\nt,a_0,u_1,\ldots,u_\ip,\la)^{-1}\cdot\mm(\nt,u_0,u_1,\ldots,u_\ip,\la)\cdot
\sho^{-1}\big(\mm(\nt,a_0,u_1,\ldots,u_\ip,\la)\big),
\end{multline}
where the right-hand side of~\er{hmmm} does not depend on~$u_\ip$.

Furthermore, the above results imply the following additional statements\textup{:}

Let $(\mm,\lt)$ be a MLR \ofc{0}{\ip}. Then \er{pdm} is a necessary and sufficient condition for
existence of a gauge transformation $\gbf=\gbf(\nt,\ud,\la)$
such that the corresponding gauge equivalent MLR
\begin{gather}
\lb{hm2}
\widehat \mm=\sho(\gbf)\cdot \mm\cdot\gbf^{-1},\qquad\quad
\widehat\lt=\tdt(\gbf)\cdot\gbf^{-1}+\gbf\cdot\lt\cdot\gbf^{-1}
\end{gather}
is \ofc{0}{\ip\!-\!1}. If\/ 
$\mm$ satisfies~\er{pdm} then one can take 
the gauge transformation~\er{guum} for this purpose.
\end{theorem}
\begin{proof}
Let us show that \textbf{Property~1} implies \textbf{Property~2}.
Suppose that there is a gauge transformation~$\gbf=\gbf(\nt,\ud,\la)$ 
such that the matrix~\er{hmgm} is of the form~\er{hmhm}. We need to prove~\er{pdm}.

Let us first prove the following
\begin{gather}
\lb{dkg}
\forall\,q\ge\ip\qquad\forall\,j\in\{1,\ldots,\ndu\}\qquad\quad
\dd{u^j_q}(\gbf)=0.
\end{gather}
Suppose that~\er{dkg} does not hold.
Let $k$ be the maximal integer such that 
there is $j\in\{1,\ldots,\ndu\}$ satisfying $\dd{u^j_k}(\gbf)\neq 0$.
By our assumption, $k\ge\ip$, and one has $\dd{u^{\tilde{\jmath}}_{k+1}}(\gbf)=0$ 
for all $\tilde{\jmath}\in\{1,\ldots,\ndu\}$. From~\er{hmgm},~\er{hmhm} we get
\begin{gather}
\lb{csgbf}
\sho(\gbf)=\widehat{\mm}(\nt,u_0,\ldots,u_{\ip-1},\la)\cdot\gbf\cdot\big(\mm(\nt,u_0,\ldots,u_\ip,\la)\big)^{-1},\\
\lb{gbfeq}
\gbf=\big(\widehat{\mm}(\nt,u_0,\ldots,u_{\ip-1},\la)\big)^{-1}
\cdot\sho(\gbf)\cdot\mm(\nt,u_0,\ldots,u_\ip,\la).
\end{gather}
Using the rules~\er{csuf} and the property $\dd{u^j_k}(\gbf)\neq 0$,
for the left-hand side of equation~\er{csgbf} 
we derive $\dd{u^j_{k+1}}\big(\sho(\gbf)\big)=\sho\big(\dd{u^j_{k}}(\gbf)\big)\neq 0$.
On the other hand, since $k\ge\ip$ and $\dd{u^{j}_{k+1}}(\gbf)=0$,
the right-hand side of equation~\er{csgbf} does not depend on~$u^{j}_{k+1}$.
The obtained contradiction allows us to conclude that \er{dkg} actually holds.

Now let us prove that
\begin{gather}
\lb{dq0g}
\forall\,q\le -1\qquad\forall\,j\in\{1,\ldots,\ndu\}\qquad\quad
\dd{u^j_q}(\gbf)=0.
\end{gather}
Suppose that~\er{dq0g} does not hold.
Let $r$ be the minimal integer such that 
there is $j\in\{1,\ldots,\ndu\}$ satisfying $\dd{u^j_r}(\gbf)\neq 0$.
By our assumption, $r\le -1$, and one has $\dd{u^{\tilde{\jmath}}_{r-1}}(\gbf)=0$ 
for all $\tilde{\jmath}\in\{1,\ldots,\ndu\}$.
Therefore, the right-hand side of~\er{gbfeq} does not depend on~$u^j_r$.
On the other hand, for the left-hand side of~\er{gbfeq} we have $\dd{u^j_r}(\gbf)\neq 0$.
This contradiction shows that \er{dq0g} actually holds.

The obtained properties~\er{dkg},~\er{dq0g} imply that $\gbf$ is of the form 
$\gbf=\gbf(\nt,u_0,\dots,u_{\ip-1},\la)$. Then from~\er{hmgm} one gets
\begin{gather}
\lb{mhm}
\mm=\sho\big(\gbf(\nt,u_0,\dots,u_{\ip-1},\la)^{-1}\big)\cdot
\widehat{\mm}(\nt,u_0,\dots,u_{\ip-1},\la)\cdot\gbf(\nt,u_0,\dots,u_{\ip-1},\la).
\end{gather}
Using~\er{mhm}, for each $j=1,\dots,\ndu$ we derive
\begin{multline}
\lb{dujmg}
\dd{u^j_\ip}(\mm)\cdot\mm^{-1}=
\dd{u^j_\ip}\Big(\sho\big(\gbf(\nt,u_0,\dots,u_{\ip-1},\la)^{-1}\big)\cdot
\widehat{\mm}(\nt,u_0,\dots,u_{\ip-1},\la)\cdot\gbf(\nt,u_0,\dots,u_{\ip-1},\la)\Big)\cdot\mm^{-1}=\\
=\sho\Big(\dd{u^j_{\ip-1}}\big(\gbf(\nt,u_0,\dots,u_{\ip-1},\la)^{-1}\big)\Big)\cdot
\widehat{\mm}\cdot\gbf\cdot\big(\sho(\gbf^{-1})\cdot
\widehat{\mm}\cdot\gbf\big)^{-1}=\\
=\sho\Big(\dd{u^j_{\ip-1}}\big(\gbf(\nt,u_0,\dots,u_{\ip-1},\la)^{-1}\big)\Big)\cdot
\widehat{\mm}\cdot\gbf\cdot\gbf^{-1}\cdot\widehat{\mm}^{-1}\cdot
\big(\sho(\gbf^{-1})\big)^{-1}=\\
=\sho\Big(\dd{u^j_{\ip-1}}\big(\gbf(\nt,u_0,\dots,u_{\ip-1},\la)^{-1}\big)\Big)\cdot
\sho\big(\gbf(\nt,u_0,\dots,u_{\ip-1},\la)\big).
\end{multline}
The rules~\er{csuf} imply that the right-hand side of~\er{dujmg} does not depend on~$u^i_0$ 
for any $i=1,\dots,\ndu$. This yields~\er{pdm}.

Now let us prove that \textbf{Property~2} implies \textbf{Property~1},
along with the statement about the gauge transformation~\er{guum}.

Suppose that $\mm=\mm(\nt,u_0,u_1,\dots,u_\ip,\la)$ satisfies~\er{pdm}. 
According to Remark~\ref{rchv}, 
we can choose a constant vector $a_0\in\fik^\ndu$ such that
$\mm(\nt,a_0,u_1,\dots,u_\ip,\la)^{-1}$ is well defined.
This allows us to consider the gauge transformation~\er{guum} 
and the corresponding matrix~$\widehat{\mm}$ determined by~\er{hmgm}. 
Since for~\er{guum} one has
\begin{gather*}
\sho(\gbf)=\mm(\nt,a_0,u_1,\dots,u_\ip,\la)^{-1},\qquad\quad
\gbf^{-1}=\sho^{-1}\big(\mm(\nt,a_0,u_1,\dots,u_\ip,\la)\big),
\end{gather*}
we obtain~\er{hmmm}, which implies that the matrix-function~$\widehat{\mm}$
may depend only on $\nt,u_0,u_1,\dots,u_\ip,\la$.
In order to show that actually $\widehat{\mm}$ is of the form~\er{hmhm}, it remains to prove that 
\begin{gather}
\lb{dujphm}
\dd{u^j_\ip}\big(\widehat{\mm}\big)=0\qquad\quad\forall\, j=1,\dots,\ndu.
\end{gather}

Let $j\in\{1,\dots,\ndu\}$. Set 
$\mathbf{H}:=\gbf^{-1}(\nt,u_0,\dots,u_{\ip-1},\la)=\sho^{-1}\big(\mm(\nt,a_0,u_1,\dots,u_\ip,\la)\big)$.
Using~\er{hmmm} and the property $\dd{u^j_\ip}(\mathbf{H})=0$, we get
\begin{multline}
\lb{pdhatm}
\dd{u^j_\ip}\big(\widehat\mm\big)
=\dd{u^j_\ip}\Big(\mm(\nt,a_0,u_1,\dots,u_\ip,\la)^{-1}\cdot \mm(\nt,u_0,u_1,\dots,u_\ip,\la)\Big)
\cdot\mathbf{H}=\\
=\Big(\dd{u^j_\ip}\big(\mm(\nt,a_0,u_1,\dots,u_\ip,\la)^{-1}\big)\cdot \mm(\nt,u_0,u_1,\dots,u_\ip,\la)+\\
+\mm(\nt,a_0,u_1,\dots,u_\ip,\la)^{-1}\cdot\dd{u^j_\ip}\big(\mm(\nt,u_0,u_1,\dots,u_\ip,\la)\big)\Big)
\cdot\mathbf{H}=\\
=\Big(-\mm(\nt,a_0,u_1,\dots,u_\ip,\la)^{-1}\cdot\dd{u^j_\ip}\big(\mm(\nt,a_0,u_1,\dots,u_\ip,\la\big)\cdot 
\mm(\nt,a_0,u_1,\dots,u_\ip,\la)^{-1}\cdot\\
\cdot\mm(\nt,u_0,u_1,\dots,u_\ip,\la)+
\mm(\nt,a_0,u_1,\dots,u_\ip,\la)^{-1}\cdot\dd{u^j_\ip}\big(\mm(\nt,u_0,u_1,\dots,u_\ip,\la)\big)\Big)
\cdot\mathbf{H}=\\
=\mm(\nt,a_0,u_1,\dots,u_\ip,\la)^{-1}\cdot\Big(-\dd{u^j_\ip}\big(\mm(\nt,a_0,u_1,\dots,u_\ip,\la)\big)\cdot
\mm(\nt,a_0,u_1,\dots,u_\ip,\la)^{-1}+\\
+\dd{u^j_\ip}\big(\mm(\nt,u_0,u_1,\dots,u_\ip,\la)\big)\cdot
\mm(\nt,u_0,u_1,\dots,u_\ip,\la)^{-1}\Big)\cdot \mm(\nt,u_0,u_1,\dots,u_\ip,\la)\cdot\mathbf{H}=\\
=\mm(\nt,a_0,u_1,\dots,u_\ip,\la)^{-1}\cdot \mathbf{L}^j(\nt,u_0,u_1,\dots,u_\ip,\la)\cdot \mm(\nt,u_0,u_1,\dots,u_\ip,\la)\cdot\mathbf{H},
\end{multline}
where 
\begin{multline}
\lb{Lu}
\mathbf{L}^j(\nt,u_0,u_1,\dots,u_\ip,\la)=-\dd{u^j_\ip}\big(\mm(\nt,a_0,u_1,\dots,u_\ip,\la)\big)\cdot
\mm(\nt,a_0,u_1,\dots,u_\ip,\la)^{-1}+\\
+\dd{u^j_\ip}\big(\mm(\nt,u_0,u_1,\dots,u_\ip,\la)\big)\cdot \mm(\nt,u_0,u_1,\dots,u_\ip,\la)^{-1}.
\end{multline}
From~\er{Lu} and~\er{pdm} it follows that
\begin{gather}
\lb{pdL}
\begin{gathered}
\dd{u^i_0}\big(\mathbf{L}^j(\nt,u_0,u_1,\dots,u_\ip,\la)\big)=
\dd{u^i_0}\Big(\dd{u^j_\ip}
\big(\mm(\nt,u_0,u_1,\dots,u_\ip,\la)\big)\cdot \mm(\nt,u_0,u_1,\dots,u_\ip,\la)^{-1}\Big)=0\\
\forall\, i=1,\dots,\ndu,
\end{gathered}\\
\lb{La0}
\mathbf{L}^j(\nt,a_0,u_1,\dots,u_\ip,\la)=0.
\end{gather}
Since, according to Remark~\ref{rratf}, the entries of 
the matrix-function $\mathbf{L}^j(\nt,u_0,u_1,\dots,u_\ip,\la)$
depend rationally on $\nt\in\zz$, $u_0,u_1,\dots,u_\ip\in\bC^\ndu$ 
and meromorphically on~$\la\in\bC$,
equations~\er{pdL},~\er{La0} imply that the matrix-function
$\mathbf{L}^j(\nt,u_0,u_1,\dots,u_\ip,\la)$ is identically zero.

Substituting $\mathbf{L}^j(\nt,u_0,u_1,\dots,u_\ip,\la)=0$ in~\er{pdhatm}, one gets~\er{dujphm}.
Therefore, the matrix-function~$\widehat{\mm}$ given by~\er{hmgm},~\er{hmmm} is of the form~\er{hmhm},
and the right-hand side of~\er{hmmm} does not depend on~$u_\ip$. 

Clearly, the obtained results immediately imply the additional statements (about MLRs) in the end of the theorem.
\end{proof}

In what follows, we use the notation and formulas from Definition~\ref{dwd}.
\begin{lemma}
\label{lem1}
Let $M=M(\nt,\ud,\la)$, $\,\gbf=\gbf(\nt,\ud,\la)$ 
be invertible $\sm\times\sm$ matrix-functions 
and $P=P(\nt,\ud,\la)$ be an arbitrary $\sm\times\sm$ matrix-function.
Set $\widetilde{M}=\sho(\gbf)\cdot M\cdot\gbf^{-1}$.
Then 
\begin{gather}
\lb{tmm}
\nab_{\widetilde{M}}^{\km}\big(\sho(\gbf)\cdot P\cdot \sho(\gbf)^{-1}\big)=
\sho(\gbf)\cdot \nab^{\km}_M(P)\cdot \sho(\gbf)^{-1}\qquad\quad
\forall\,{\km}\in\mathbb{Z}.
\end{gather}
\end{lemma}
\begin{proof}
We have
\begin{gather}
\lb{tmtm}
\widetilde{M}=\sho(\gbf)\cdot M\cdot\gbf^{-1},\qquad\quad
\widetilde{M}^{-1}=\gbf\cdot M^{-1}\cdot\sho(\gbf)^{-1},\qquad\quad
\sho(\gbf^{-1})=\sho(\gbf)^{-1}.
\end{gather}
Using~\er{nbmq} and~\er{tmtm}, one gets
\begin{multline}
\notag
\nab_{\widetilde{M}}\big(\sho(\gbf)\cdot P\cdot \sho(\gbf)^{-1}\big)=
\sho\big(\widetilde{M}^{-1}\sho(\gbf)P\sho(\gbf)^{-1}\widetilde{M}\big)=
\sho\big(\gbf M^{-1}\sho(\gbf)^{-1}\sho(\gbf)
P\sho(\gbf)^{-1}\sho(\gbf)M\gbf^{-1}\big)=\\
=\sho(\gbf M^{-1}PM\gbf^{-1})=
\sho(\gbf)\sho(M^{-1}PM)\sho(\gbf^{-1})=
\sho(\gbf)\nab_M(P)\sho(\gbf)^{-1}.
\end{multline}
Thus we have
\begin{gather}
\lb{nbtm}
\nab_{\widetilde{M}}\big(\sho(\gbf)\cdot P\cdot \sho(\gbf)^{-1}\big)=
\sho(\gbf)\cdot \nab_M(P)\cdot \sho(\gbf)^{-1}.
\end{gather}
Similarly, using~\er{nbm1q} and~\er{tmtm}, one obtains
\begin{gather}
\lb{nbtmm1}
\nab_{\widetilde{M}}^{-1}\big(\sho(\gbf)\cdot P\cdot \sho(\gbf)^{-1}\big)=
\sho(\gbf)\cdot \nab_M^{-1}(P)\cdot \sho(\gbf)^{-1}.
\end{gather}
Using~\er{nbtm} and~\er{nbtmm1}, one can prove~\er{tmm} by induction on~${\km}$.
\end{proof}

\begin{theorem}
\label{th1}
For any invertible $\sm\times\sm$ matrix-functions
$M=M(\nt,\ud,\la)$ and $G=G(\nt,\ud,\la)$, we have
\begin{gather}
\label{gmg}
\euop{u^j}\big(\sho(G)\cdot M\cdot G^{-1}\big)=\sho(G)\cdot\euop{u^j}(M)\cdot\sho(G)^{-1}
\qquad\quad\forall\,j=1,\ldots,\ndu.
\end{gather}
\end{theorem}
\begin{proof}
Set $\widetilde{M}=\sho(G)MG^{-1}$.
To prove~\er{gmg}, we need to show that 
$\euop{u^j}\big(\widetilde{M}\big)=\sho(G)\cdot\euop{u^j}(M)\cdot\sho(G)^{-1}$.

Using the relations
\begin{gather*}
\widetilde{M}^{-1}=G M^{-1}\sho(G)^{-1},\qquad 
\dd{u^j_\ik}(G^{-1})=-G^{-1}\dd{u^j_\ik}(G)G^{-1},\qquad
\dd{u^j_\ik}\big(\sho(G)\big)=\sho\big(\dd{u^j_{\ik-1}}(G)\big),\\
\sho\big(\dd{u^j_{\ik-1}}(G)\big)=\sho(G)\sho\big(G^{-1}\dd{u^j_{\ik-1}}(G)\big),\qquad
MG^{-1}\dd{u^j_\ik}(G)M^{-1}=\nab_{M}^{-1}\big(\sho\big(G^{-1}\dd{u^j_\ik}(G)\big)\big)
\end{gather*}
and Lemma~\ref{lem1}, one can prove 
$\euop{u^j}\big(\widetilde{M}\big)=\sho(G)\cdot\euop{u^j}(M)\cdot\sho(G)^{-1}$
by straightforward (but rather lengthy) computations,
which are presented in the preprint~\cite{gaprep2026}.
\end{proof}


\begin{lemma}
\label{leop0}
Let $\ip\in\zp$. 
Consider an invertible $\sm\times\sm$ matrix-function $M=M(\nt,u_0,\ldots,u_\ip,\la)$. 
Suppose that 
\begin{gather}
\lb{eop0}
\euop{u^j}(M)=0\qquad\quad\forall\,j=1,\ldots,\ndu.
\end{gather}
Then there is a gauge transformation $\gbf=\gbf(\nt,u_0,\ldots,u_{\ip-1},\la)$
such that the matrix
\begin{gather}
\lb{hatmp}
\breve{M}=\sho(\gbf)\cdot M(\nt,u_0,\ldots,u_\ip,\la)\cdot\gbf^{-1}
\end{gather}
obeys 
\begin{gather}
\lb{dujhm}
\dd{u^j_\ik}(\breve{M})=0\qquad\quad\forall\,\ik\in\zz,\qquad\forall\,j=1,\ldots,\ndu.
\end{gather}
\end{lemma}
\begin{proof}
We prove this by induction on~$\ip\in\zp$. First, consider the case $\ip=0$, where $M=M(\nt,u_0,\la)$.
According to~\er{foreo}, for $M=M(\nt,u_0,\la)$ we have $\euop{u^j}(M)=\dd{u^j_0}(M)\cdot M^{-1}$. 
Then \er{eop0} implies that $M$ is of the form $M=M(\nt,\la)$.
We take $\gbf=\id$, where $\id$ is the identity matrix of size~$\sm\times\sm$. 
Then $\breve{M}=M(\nt,\la)$ satisfies~\er{hatmp},~\er{dujhm}.

To make the induction step, we consider a number $k\in\zp$ such that 
the statement of Lemma~\ref{leop0} is valid for~$\ip=k$.
We need to prove the statement in the case $\ip=k+1$, where $M=M(\nt,u_0,\ldots,u_{k+1},\la)$.
According to~\er{foreo}, for such~$M$ equation~\er{eop0} reads
$\sum_{\ik=0,\ldots,k+1}\nab_M^{-\ik}\big(\dd{u^j_\ik}(M)M^{-1}\big)=0$, which yields
\begin{gather}
\lb{k1}
\nab_M^{-(k+1)}\big(\dd{u^j_{k+1}}(M)M^{-1}\big)=
-\sum_{\ik=0,\ldots,k}\nab_M^{-\ik}\big(\dd{u^j_\ik}(M)M^{-1}\big)
\qquad\quad\forall\,j=1,\ldots,\ndu.
\end{gather}
Applying the operator~$\nab_M^{k+1}$ to equation~\er{k1} and using~\er{nbmq}, we derive
\begin{multline}
\lb{k1new}
\dd{u^j_{k+1}}(M)\cdot M^{-1}=
-\sum_{\ik=0,\ldots,k}\nab_M^{k+1-\ik}\big(\dd{u^j_\ik}(M)M^{-1}\big)=
-\nab_M\Big(\sum_{\ik=0,\ldots,k}\nab_M^{k-\ik}\big(\dd{u^j_\ik}(M)M^{-1}\big)\Big)=\\
=-\sho\Big(M^{-1}\Big(\sum_{\ik=0,\ldots,k}\nab_M^{k-\ik}\big(\dd{u^j_\ik}(M)M^{-1}\big)\Big)M\Big)
\qquad\quad\forall\,j=1,\ldots,\ndu.
\end{multline}
Let $Y=M^{-1}\Big(\sum_{\ik=0,\ldots,k}\nab_M^{k-\ik}\big(\dd{u^j_\ik}(M)M^{-1}\big)\Big)M$.
The formulas $M=M(\nt,u_0,\ldots,u_{k+1},\la)$ and~\er{nbmq}
imply that $Y$ does not depend on~$u^i_a$ for any $a<0$ and $i=1,\ldots,\ndu$. 
Since the right-hand side of~\er{k1new} is equal to~$-\sho(Y)$, 
it does not depend on~$u^i_0$, $\,i=1,\ldots,\ndu$.
Therefore, from~\er{k1new} we get 
\begin{gather}
\lb{pdmnew}
\forall\, i,j=1,\ldots,\ndu\qquad\quad
\dd{u^i_0}\Big(\dd{u^j_{k+1}}\big(M(\nt,u_0,\ldots,u_{k+1},\la)\big)\cdot M(\nt,u_0,\ldots,u_{k+1},\la)^{-1}\Big)=0.
\end{gather}
By Theorem~\ref{thmlpk1} in the case $\ip=k+1$, 
from property~\er{pdmnew} it follows that there is a gauge transformation $G=G(\nt,u_0,\ldots,u_{k},\la)$
such that the matrix
\begin{gather}
\lb{lhm}
\widehat{M}=\sho(G)\cdot M(\nt,u_0,\ldots,u_{k+1},\la)\cdot G^{-1}
\end{gather}
is of the form $\widehat{M}=\widehat{M}(\nt,u_0,\ldots,u_k,\la)$.
Using~\er{gmg} and~\er{eop0}, one obtains
\begin{gather}
\label{gmgnew}
\euop{u^j}\big(\widehat{M}\big)=\euop{u^j}\big(\sho(G)\cdot M\cdot G^{-1}\big)
=\sho(G)\cdot\euop{u^j}(M)\cdot\sho(G)^{-1}=0
\qquad\quad\forall\,j=1,\ldots,\ndu.
\end{gather}
Since, by the induction assumption, the statement of Lemma~\ref{leop0} is valid for~$\ip=k$, 
equation~\er{gmgnew} for $\widehat{M}=\widehat{M}(\nt,u_0,\ldots,u_k,\la)$ implies that 
there is a gauge transformation $\gbb=\gbb(\nt,u_0,\ldots,u_{k-1},\la)$ such that the matrix
\begin{gather}
\lb{hatmp2}
\breve{M}=\sho(\gbb)\cdot\widehat{M}(\nt,u_0,\ldots,u_k,\la)\cdot\gbb^{-1}
\end{gather}
obeys~\er{dujhm}. Now we set 
\begin{gather}
\lb{ggg}
\gbf(\nt,u_0,\ldots,u_k,\la)=\gbb(\nt,u_0,\ldots,u_{k-1},\la)\cdot G(\nt,u_0,\ldots,u_{k},\la).
\end{gather}
Formulas~\er{lhm},~\er{hatmp2},~\er{ggg} imply that $\breve{M}$ satisfies~\er{hatmp} 
with~$\ip=k+1$. Furthermore, as said above, $\breve{M}$ obeys~\er{dujhm}.
Therefore, we have proved the statement of Lemma~\ref{leop0} in the case~$\ip=k+1$.
\end{proof}

\begin{theorem}
\label{thtriv}
A MLR $\big(\mm(n,\ud,\la),\,\lt(n,\ud,\la)\big)$ 
is gauge equivalent to a trivial MLR if and only if\/ 
$\mm=\mm(n,\ud,\la)$ satisfies
\begin{gather}
\lb{eomm0}
\euop{u^j}(\mm)=0\qquad\quad\forall\,j=1,\ldots,\ndu,
\end{gather}
where $\euop{u^j}(\mm)$ is determined by formula~\er{foreo}.
\end{theorem}
\begin{proof}
Suppose that a MLR $\big(\mm(n,\ud,\la),\,\lt(n,\ud,\la)\big)$
is gauge equivalent to a trivial MLR. This means that there is a 
gauge transformation $\gbf=\gbf(\nt,\ud,\la)$ such that the matrix
\begin{gather}
\lb{hatmp3}
\widehat{\mm}=\sho(\gbf)\cdot\mm(n,\ud,\la)\cdot\gbf^{-1}
\end{gather}
obeys 
\begin{gather}
\lb{duhmm}
\dd{u^j_\ik}(\widehat{\mm})=0\qquad\quad\forall\,\ik\in\zz,\qquad\forall\,j=1,\ldots,\ndu.
\end{gather}
According to~\er{foreo}, property~\er{duhmm} yields
\begin{gather}
\lb{eohm0}
\euop{u^j}(\widehat{\mm})=0\qquad\quad\forall\,j=1,\ldots,\ndu.
\end{gather}
From~\er{hatmp3} one derives 
\begin{gather}
\lb{mmghm}
\mm=\sho(\gbf)^{-1}\cdot\widehat{\mm}\cdot\gbf=
\sho\big(\gbf^{-1}\big)\cdot\widehat{\mm}\cdot\gbf.
\end{gather}
Using~\er{mmghm},~\er{eohm0} and property~\er{gmg} for $M=\widehat{\mm}$ and $G=\gbf^{-1}$, we get
\begin{gather}
\notag
\euop{u^j}(\mm)=\euop{u^j}\big(\sho\big(\gbf^{-1}\big)\cdot\widehat{\mm}\cdot\gbf\big)=
\sho\big(\gbf^{-1}\big)\cdot\euop{u^j}(\widehat{\mm})\cdot\sho(\gbf)=0\qquad\quad\forall\,j=1,\ldots,\ndu,
\end{gather}
thus we obtain~\er{eomm0}.

Now, conversely, consider a MLR $(\mm,\lt)$
with $\sho$-part $\mm=\mm(n,\ud,\la)$ satisfying~\er{eomm0}.
We need to prove that this MLR is gauge equivalent to a trivial MLR. 

The rules~\er{csuf} imply that there are $\imm\in\zz$ and $\ip\in\zp$ such that
the matrix-function $\widetilde{\mm}=\sho^\imm(\mm)$ may depend only on $n,u_0,\ldots,u_\ip,\la$.
By Lemma~\ref{lgs}, the MLR $(\mm,\lt)$ 
is gauge equivalent to some MLR $\big(\widetilde{\mm},\widetilde{\lt}\big)$ with $\sho$-part equal to 
$\sho^\imm(\mm)=\widetilde{\mm}(n,u_0,\ldots,u_\ip,\la)$.
Therefore, there is a gauge transformation $G=G(n,\ud,\la)$ 
such that $\widetilde{\mm}=\sho(G)\cdot\mm\cdot G^{-1}$.
Using~\er{gmg} and~\er{eomm0}, we get
\begin{gather}
\label{eotm}
\euop{u^j}\big(\widetilde{\mm}\big)=
\euop{u^j}\big(\sho(G)\cdot\mm\cdot G^{-1}\big)=\sho(G)\cdot\euop{u^j}(\mm)\cdot\sho(G)^{-1}=0
\qquad\quad\forall\,j=1,\ldots,\ndu.
\end{gather}
By Lemma~\ref{leop0} for $\widetilde{\mm}=\widetilde{\mm}(n,u_0,\ldots,u_\ip,\la)$, 
property~\er{eotm} implies that there is a gauge transformation $\gbf=\gbf(u_0,\ldots,u_{\ip-1},\la)$
such that the matrix $\breve{M}=\sho(\gbf)\cdot\widetilde{\mm}\cdot\gbf^{-1}$ obeys~\er{dujhm}.

The relations $\breve{M}=\sho(\gbf)\cdot\widetilde{\mm}\cdot\gbf^{-1}$ and 
$\widetilde{\mm}=\sho(G)\cdot\mm\cdot G^{-1}$ yield 
$\breve{M}=\sho(\gbf G)\cdot\mm\cdot(\gbf G)^{-1}$. 
Therefore, applying the gauge transformation $\gbf G$ to the MLR~$(\mm,\lt)$,
we get a MLR with $\sho$-part~$\breve{M}$ obeying~\er{dujhm}.
Hence $(\mm,\lt)$ is gauge equivalent to a trivial MLR.
\end{proof}

\section{Miura-type transformations and matrix Lax representations}
\lb{smt}

The following lemma is proved easily by a direct calculation.
\begin{lemma}
\label{luv}
Suppose that matrices
\begin{gather}
\lb{mltu}
\mm=\mm(\nt,\ud,\la),\qquad\quad\lt=\lt(\nt,\ud,\la)
\end{gather}
form a MLR for equation~\eqref{sdde}.
Let \er{uvf} be a MT from equation~\eqref{vdde} to equation~\eqref{sdde}.
As explained in Definition~\textup{\ref{defmtt}}, in formula~\er{uvf} one has
an $\ndu$-component vector-function $\mf=(\mf^1,\dots,\mf^{\ndu})$ 
which  may depend on~$n$ and on a finite number of the dynamical variables 
$\mv_\ik=(\mv^1_\ik,\dots,\mv^\ndu_\ik)$, $\,\ik\in\zz$.

Motivated by formula~\er{uvf}, we substitute $u_\ik=\sho^\ik(\mf)$, $\,\ik\in\zz$, 
in~\er{mltu} and obtain the matrices
\begin{gather}
\lb{mlmv}
\mm=\mm\big(\nt,\sho^\ik(\mf),\ldots,\la\big),
\qquad\quad\lt=\lt\big(\nt,\sho^\ik(\mf),\ldots,\la\big).
\end{gather}
Then \er{mlmv} is a MLR for equation~\eqref{vdde}.
\end{lemma}

\subsection{Non-autonomous examples}
\lb{snaex}

The following $2$-component non-autonomous equation
was introduced in~\cite[page~7]{KMX2015} in a study of Darboux transformations
\begin{gather}
\label{kmx}
\left\{
\begin{aligned}
\pd_t(\pe^1)&=2\big(\pe^{2}_{1}(\pe^{1}_{0})^2+\kmx(n)\pe^{1}_{0}-\pe^{1}_{1}\big),\\ 
\pd_t(\pe^2)&=2\big(\pe^{2}_{-1}-\pe^{1}_{-1}(\pe^{2}_{0})^2-\kmx(n-1)\pe^{2}_{0}\big).
\end{aligned}\right.
\end{gather}
Here $\kmx(n)$ is an arbitrary function.
(This equation is given in~\cite[page~7]{KMX2015} in different notation,
using functions $p(n,x)$, $q(n,x)$, instead of $\pe^1(n,t)$, $\pe^2(n,t)$.)

The results of~\cite{KMX2015} imply 
that the following matrices form a MLR for~\er{kmx}
\begin{gather}
\lb{mlrbc}
\widetilde{\mm}=
\begin{pmatrix}
 {\la}+\kmx(n)+\pe^{1}_{0} \pe^{2}_{1} & \pe^{1}_{0} \\
 \pe^{2}_{1} & 1 \\
\end{pmatrix},
\qquad\quad\widetilde{\lt}=
\begin{pmatrix}
 -{\la} & -2 \pe^{1}_{0} \\
 -2 \pe^{2}_{0} & {\la} \\
\end{pmatrix}.
\end{gather}
The MLR~\er{mlrbc} is not $u$-tame (in the sense of Definition~\ref{dtame}).
In order to obtain a $u$-tame MLR, we relabel 
\begin{gather}
\lb{rp}
\pe^1_\ik\,\mapsto\,\pe^1_\ik,\qquad\quad \pe^2_\ik\,\mapsto\,\pe^2_{\ik-1}\qquad\quad\forall\,\ik\in\zz.
\end{gather}
This means that we make the following invertible change of variables 
\begin{gather}
\notag
\pe^1(n,t)\,\mapsto\,\pe^1(n,t),\qquad\pe^2(n,t)\,\mapsto\,\pe^2(n\!-\!1,t).
\end{gather}
Applying the (invertible) change~\er{rp} to \er{kmx}, \er{mlrbc}, we get
\begin{gather}
\label{chkmx}
\left\{
\begin{aligned}
\pd_t(\pe^1)&=2\big(\pe^{2}_{0}(\pe^{1}_{0})^2+\kmx(n)\pe^{1}_{0}-\pe^{1}_{1}\big),\\ 
\pd_t(\pe^2_{-1})&=2\big(\pe^{2}_{-2}-\pe^{1}_{-1}(\pe^{2}_{-1})^2-\kmx(n-1)\pe^{2}_{-1}\big),
\end{aligned}\right.\\
\lb{nmlrbc}
\breve{\mm}=
\begin{pmatrix}
 {\la}+\kmx(n)+\pe^{1}_{0} \pe^{2}_{0} & \pe^{1}_{0} \\
 \pe^{2}_{0} & 1 \\
\end{pmatrix},
\qquad\quad
\breve{\lt}=
\begin{pmatrix}
 -{\la} & -2 \pe^{1}_{0} \\
 -2 \pe^{2}_{-1} & {\la} \\
\end{pmatrix}.
\end{gather}
Applying the (invertible) operator~$\sho$ to the second component in~\er{chkmx}, one obtains
\begin{gather}
\label{newkmx}
\left\{
\begin{aligned}
\pd_t(\pe^1)&=2\big(\pe^{2}_{0}(\pe^{1}_{0})^2+\kmx(n)\pe^{1}_{0}-\pe^{1}_{1}\big),\\ 
\pd_t(\pe^2)&=2\big(\pe^{2}_{-1}-\pe^{1}_0(\pe^{2}_0)^2-\kmx(n)\pe^{2}_0\big).
\end{aligned}\right.
\end{gather}
Since \er{mlrbc} is a MLR for equation~\er{kmx},
the described way of deriving~\er{nmlrbc},~\er{newkmx}
from~\er{mlrbc},~\er{kmx} implies that \er{nmlrbc} is a MLR for equation~\er{newkmx}.
Furthermore, equation~\er{newkmx} is equivalent to equation~\er{kmx},
as \er{newkmx} is obtained from~\er{kmx} by means of invertible operations, as described above.
The MLR~\er{nmlrbc} is $u$-tame, and this will allow us to construct a MT, as shown below.

Fix a constant $\mc\in\mathbb{C}$.
In what follows, we use the matrices~\er{nmlrbc} with~$\la$ replaced by~$\mc$.
According to Definition~\ref{dmlpgt}, since the matrices~\er{nmlrbc} 
form a MLR for equation~\er{newkmx}, we can consider the auxiliary system
\begin{gather}
\lb{syspsit}
\sho(\Psi)=\begin{pmatrix}
 {\mc}+\kmx(n)+\pe^{1}_{0} \pe^{2}_{0} & \pe^{1}_{0} \\
 \pe^{2}_{0} & 1 \\
\end{pmatrix}\cdot\Psi,\qquad\quad
\pd_t(\Psi)=\begin{pmatrix}
 -{\mc} & -2 \pe^{1}_{0} \\
 -2 \pe^{2}_{-1} & {\mc} \\
\end{pmatrix}\cdot\Psi,
\end{gather}
which is compatible modulo equation~\eqref{newkmx}. 
Here $\Psi=\Psi(n,t)$ is an invertible $2\times 2$ matrix
whose elements we denote by $\uppsi^{ij}=\uppsi^{ij}(n,t)$, $\,i,j=1,2$.

Substituting $\Psi=\begin{pmatrix}
\uppsi^{11}&\uppsi^{12}\\
\uppsi^{21}&\uppsi^{22}
\end{pmatrix}$ in~\er{syspsit} and considering the first columns 
in the matrix equations~\er{syspsit}, one gets
\begin{gather}
\lb{sp11}
\sho(\uppsi^{11})=({\mc}+\kmx(n)+\pe^{1}_{0} \pe^{2}_{0})\uppsi^{11}
+\pe^{1}_{0}\uppsi^{21},\qquad\quad
\sho(\uppsi^{21})=\pe^{2}_{0}\uppsi^{11}+\uppsi^{21},\\
\lb{dp11}
\pd_t(\uppsi^{11})=-{\mc}\uppsi^{11}-2 \pe^{1}_{0}\uppsi^{21},\qquad\quad
\pd_t(\uppsi^{21})=-2 \pe^{2}_{-1}\uppsi^{11}+{\mc}\uppsi^{21}.
\end{gather}
Now we employ the standard idea to consider the function $\vv={\uppsi^{11}}/{\uppsi^{21}}$.
Using~\er{sp11},~\er{dp11}, we derive
\begin{gather}
\label{swp11}
\begin{split}
\sho(\vv)&=\frac{\sho(\uppsi^{11})}{\sho(\uppsi^{21})}=
\frac{({\mc}+\kmx(n)+\pe^{1}_{0} \pe^{2}_{0})\uppsi^{11}+\pe^{1}_{0}\uppsi^{21}}%
{\pe^{2}_{0}\uppsi^{11}+\uppsi^{21}}=\\
&=\frac{({\mc}+\kmx(n)+\pe^{1}_{0} \pe^{2}_{0})(\uppsi^{11}/\uppsi^{21})+\pe^{1}_{0}}%
{\pe^{2}_{0}(\uppsi^{11}/{\uppsi^{21}})+1}=
\frac{({\mc}+\kmx(n)+\pe^{1}_{0} \pe^{2}_{0})\vv+\pe^{1}_{0}}%
{\pe^{2}_{0}\vv+1},
\end{split}\\
\label{dwp11}
\begin{split}
\pd_t(\vv)&=\pd_t\Big(\frac{\uppsi^{11}}{\uppsi^{21}}\Big)=
\frac{\pd_t(\uppsi^{11})}{\uppsi^{21}}-\frac{\uppsi^{11}\pd_t(\uppsi^{21})}{(\uppsi^{21})^2}=
\frac{-{\mc}\uppsi^{11}-2 \pe^{1}_{0}\uppsi^{21}}{\uppsi^{21}}-
\frac{\uppsi^{11}(-2 \pe^{2}_{-1}\uppsi^{11}+{\mc}\uppsi^{21})}{(\uppsi^{21})^2}=\\
&=-{\mc}\frac{\uppsi^{11}}{\uppsi^{21}}-2 \pe^{1}_{0}
+2\pe^{2}_{-1}\Big(\frac{\uppsi^{11}}{\uppsi^{21}}\Big)^2-
{\mc}\frac{\uppsi^{11}}{\uppsi^{21}}=2\pe^{2}_{-1}(\vv)^2-2{\mc}\vv-2 \pe^{1}_{0}.
\end{split}
\end{gather}
From~\er{swp11},~\er{dwp11} one gets
\begin{gather}
\label{swuu}
\sho(\vv)=\frac{({\mc}+\kmx(n)+\pe^{1}_{0} \pe^{2}_{0})\vv+\pe^{1}_{0}}{\pe^{2}_{0}\vv+1},\\
\label{dwuu}
\pd_t(\vv)=2\pe^{2}_{-1}(\vv)^2-2{\mc}\vv-2 \pe^{1}_{0}.
\end{gather}
We have derived equations~\er{swuu},~\er{dwuu} from system~\er{syspsit}, 
which is compatible modulo equation~\eqref{newkmx}. 
The described procedure of deriving~\er{swuu},~\er{dwuu} implies that  
system~\er{swuu},~\er{dwuu} is compatible modulo equation~\eqref{newkmx} as well. 

From~\er{swuu} one obtains
\begin{gather}
\label{reseq}
\pe^1=\pe^1_0=\frac{\sho(\vv) \pe^{2}_{0} \vv+\sho(\vv)-{\mc} \vv-\kmx(n) \vv}{\vv \pe^{2}_{0}+1}.
\end{gather}
We set
\begin{gather}
\label{vwu}
\mv^1:=\vv,\qquad\qquad\mv^2:=\pe^2.
\end{gather}
From~\er{vwu} one derives $\sho(\vv)=\sho(\mv^1)=\mv^1_1$.
As usual, we have $\pe^i_0=\pe^i$, $\,\mv^i_0=\mv^i$ for $i=1,2$.
Now we substitute $\vv=\mv^1_0$, $\,\sho(\vv)=\mv^1_1$, $\,\pe^2_0=\mv^2_0$ 
in~\er{reseq} and consider the obtained equation along with $\pe^2=\mv^2_0$, which gives
\begin{gather}
\label{mtkmx}
\left\{
\begin{aligned}
\pe^1&=\frac{\mv^{1}_{1} \mv^{2}_{0} \mv^{1}_{0}+\mv^{1}_{1}-{\mc} \mv^{1}_{0}-\kmx(n) \mv^{1}_{0}}{\mv^{1}_{0} \mv^{2}_{0}+1},\\ 
\pe^2&=\mv^{2}_{0}.
\end{aligned}\right.
\end{gather}
For each $\ell\in\zz$, applying the operator~$\sho^\ell$ to~\er{mtkmx}, we get
\begin{gather}
\label{uellmt}
\left\{
\begin{aligned}
\pe^1_\ell&=\frac{\mv^{1}_{\ell+1} \mv^{2}_{\ell} \mv^{1}_{\ell}+%
\mv^{1}_{\ell+1}-{\mc} \mv^{1}_{\ell}-\kmx(n+\ell) \mv^{1}_{\ell}}{\mv^{1}_{\ell} \mv^{2}_{\ell}+1},\\ 
\pe^2_\ell&=\mv^2_\ell,\qquad\qquad\qquad\qquad\ell\in\zz.
\end{aligned}\right.
\end{gather}
Substituting~$\vv=\mv^1$ and~\er{uellmt} in~\er{dwuu}, we derive
\begin{gather}
\label{dtmv1}
\pd_t(\mv^1)=2\mv^{2}_{-1}(\mv^1_0)^2-2{\mc}\mv^1_0
-\frac{2(\mv^{1}_{1} \mv^{2}_{0} \mv^{1}_{0}+\mv^{1}_{1}-{\mc} \mv^{1}_{0}-\kmx(n) \mv^{1}_{0})}%
{\mv^{1}_{0} \mv^{2}_{0}+1}.
\end{gather}
Now we consider~\er{dtmv1} along with the equation obtained 
by means of substituting~\er{uellmt} in the second component of~\er{newkmx}.
This gives the $2$-component non-autonomous equation
\begin{gather}
\label{dtmv}
\left\{
\begin{aligned}
\pd_t(\mv^1)&=2\mv^{2}_{-1}(\mv^1_0)^2-2{\mc}\mv^1_0
-\frac{2\big(\mv^{1}_{1} \mv^{2}_{0} \mv^{1}_{0}+\mv^{1}_{1}-{\mc} \mv^{1}_{0}-\kmx(n) \mv^{1}_{0}\big)}%
{\mv^{1}_{0} \mv^{2}_{0}+1},\\ 
\pd_t(\mv^2)&=2\mv^{2}_{-1}-2\kmx(n)\mv^{2}_0-\frac{2(\mv^{2}_0)^2\big(\mv^{1}_{1} \mv^{2}_{0} \mv^{1}_{0}+\mv^{1}_{1}-{\mc} \mv^{1}_{0}-\kmx(n) \mv^{1}_{0}\big)}{\mv^{1}_{0} \mv^{2}_{0}+1}.
\end{aligned}\right.
\end{gather}
As discussed above, 
system~\er{swuu},~\er{dwuu} is compatible modulo equation~\eqref{newkmx}.
Since equations~\er{mtkmx},~\er{dtmv} 
are derived from~\er{swuu},~\er{dwuu},~\er{newkmx} (by means of introducing~\er{vwu}), 
system~\er{mtkmx},~\er{dtmv} is compatible modulo equation~\eqref{newkmx} as well.
This implies that \er{mtkmx} is a MT from~\er{dtmv} to~\er{newkmx}.
(This fact can also be checked independently by a straightforward computation.)

Thus \er{mtkmx} is a MT for equation~\er{newkmx}.
As explained above, \er{newkmx} is obtained from~\er{kmx} 
by means of an invertible transformation involving~\er{rp}.
In order to return from~\er{newkmx} to~\er{kmx} and to derive a MT for equation~\eqref{kmx},
we are going to use the relabeling
\begin{gather}
\lb{invrp}
\pe^1_\ik\,\mapsto\,\pe^1_\ik,\qquad\pe^2_\ik\,\mapsto\,\pe^2_{\ik+1}\qquad\quad\forall\,\ik\in\zz,
\end{gather}
which is the inverse of~\er{rp}. 
That is, we make the change $\pe^1(n,t)\mapsto\pe^1(n,t)$, 
$\,\pe^2(n,t)\mapsto\pe^2(n\!+\!1,t)$.

Applying the relabeling~\er{invrp} to~\er{newkmx} and~\er{mtkmx}, we obtain
\begin{gather}
\label{inpeq}
\left\{
\begin{aligned}
\pd_t(\pe^1)&=2\big(\pe^{2}_{1}(\pe^{1}_{0})^2+\kmx(n)\pe^{1}_{0}-\pe^{1}_{1}\big),\\ 
\pd_t(\pe^2_1)&=2\big(\pe^{2}_{0}-\pe^{1}_0(\pe^{2}_1)^2-\kmx(n)\pe^{2}_1\big).
\end{aligned}\right.\\
\lb{imtpe}
\left\{
\begin{aligned}
\pe^1&=\frac{\mv^{1}_{1} \mv^{2}_{0} \mv^{1}_{0}+\mv^{1}_{1}-{\mc} \mv^{1}_{0}-\kmx(n) \mv^{1}_{0}}{\mv^{1}_{0} \mv^{2}_{0}+1},\\ 
\pe^2_1&=\mv^{2}_{0}.
\end{aligned}\right.
\end{gather}
Applying the operator~$\sho^{-1}$ to the second component in~\er{inpeq} and in~\er{imtpe}, 
one gets~\er{kmx} and 
\begin{gather}
\lb{mtpeqq}
\left\{
\begin{aligned}
\pe^1&=\frac{\mv^{1}_{1} \mv^{2}_{0} \mv^{1}_{0}+\mv^{1}_{1}-{\mc} \mv^{1}_{0}-\kmx(n) \mv^{1}_{0}}{\mv^{1}_{0} \mv^{2}_{0}+1},\\ 
\pe^2&=\mv^{2}_{-1}.
\end{aligned}\right.
\end{gather}
Thus \er{mtpeqq} is a MT from~\er{dtmv} to~\er{kmx}.

In what follows, we use the dynamical variables 
$\mv_\ik=(\mv^1_\ik,\mv^2_\ik)$, $\,\ik\in\zz$.
Recall that \er{nmlrbc} is a MLR for equation~\er{newkmx},
and formulas~\er{uellmt} are obtained by applying the operator~$\sho^\ell$ to~\er{mtkmx}.
Then, applying Lemma~\ref{luv} to the MLR~\er{nmlrbc} and 
to the MT~\er{mtkmx} from~\er{dtmv} to~\er{newkmx},
we substitute~\er{uellmt} in~\er{nmlrbc} and obtain the following MLR for equation~\er{dtmv}
\begin{gather}
\lb{Mbce}
\mm(n,\mv_{0},\mv_{1},\la)=\begin{pmatrix}
 {\la}+\kmx(n)+\frac{\mv^{2}_{0}(\mv^{1}_{1} \mv^{2}_{0} \mv^{1}_{0}+\mv^{1}_{1}-{\mc} \mv^{1}_{0}-\kmx(n) \mv^{1}_{0})}{\mv^{1}_{0} \mv^{2}_{0}+1} & \frac{\mv^{1}_{1} \mv^{2}_{0} \mv^{1}_{0}+\mv^{1}_{1}-{\mc} \mv^{1}_{0}-\kmx(n) \mv^{1}_{0}}{\mv^{1}_{0} \mv^{2}_{0}+1} \\
 \mv^{2}_{0} & 1 \\
\end{pmatrix},\\
\lb{Ubce}
\lt(n,\mv_{-1},\mv_{0},\mv_{1},\la)=\begin{pmatrix}
 -{\la} & -\frac{2(\mv^{1}_{1} \mv^{2}_{0} \mv^{1}_{0}+\mv^{1}_{1}-{\mc} \mv^{1}_{0}-\kmx(n) \mv^{1}_{0})}{\mv^{1}_{0} \mv^{2}_{0}+1} \\
 -2 \mv^{2}_{-1} & {\la} \\
\end{pmatrix}.
\end{gather}

In order to simplify the MLR~\er{Mbce},~\er{Ubce} by means of gauge transformations,
we consider Theorem~\ref{thmlpk1} 
in the case $\ndu=2$, $\,\sm=2$, $\,\ip=1$ with~$\mv_\ell$ instead of~$u_\ell$.
The matrix~\er{Mbce} satisfies condition~\er{pdm} with~$\ip=1$.
Therefore, we can apply Theorem~\ref{thmlpk1} to the MLR~\er{Mbce}, \er{Ubce}.
To this end, we need to do the following.
\begin{itemize}
	\item Choose a constant vector $a_0=(a^1_0,a^2_0)\in\fik^2$.
	\item Substitute $\mv_0=a_0$ in~$\mm(n,\mv_{0},\mv_{1},\la)$
given by~\er{Mbce} (i.e., substitute $\mv^i_0=a^i_0$ for $i=1,2$ in~\er{Mbce}).
	\item Compute $\gbf(n,\mv_0,\la)$ given by formula~\er{guum}, 
	which involves the matrix 
\begin{gather}
\lb{inMmbe}
\begin{gathered}
\mm(n,a_0,\mv_{1},\la)^{-1}=
\begin{pmatrix}
 {\la}+\kmx(n)+\frac{a^{2}_{0}(\mv^{1}_{1} a^{2}_{0} a^{1}_{0}+\mv^{1}_{1}-{\mc} a^{1}_{0}-\kmx(n) a^{1}_{0})}{a^{1}_{0} a^{2}_{0}+1} & \frac{\mv^{1}_{1} a^{2}_{0} a^{1}_{0}+\mv^{1}_{1}-{\mc} a^{1}_{0}-\kmx(n) a^{1}_{0}}{a^{1}_{0} a^{2}_{0}+1} \\
 a^{2}_{0} & 1 \\
\end{pmatrix}^{-1}=\\
=\begin{pmatrix}
 \frac{1}{{\la}+\kmx(n)} & \frac{a^{1}_{0} \kmx(n)+a^{1}_{0} {\mc}-a^{1}_{0} a^{2}_{0} \mv^{1}_{1}-\mv^{1}_{1}}{(a^{1}_{0} a^{2}_{0}+1) ({\la}+\kmx(n))} \\
 -\frac{a^{2}_{0}}{{\la}+\kmx(n)} & \frac{a^{1}_{0} \mv^{1}_{1} (a^{2}_{0})^2+a^{1}_{0} {\la} a^{2}_{0}-a^{1}_{0} {\mc} a^{2}_{0}+\mv^{1}_{1} a^{2}_{0}+{\la}+\kmx(n)}{(a^{1}_{0} a^{2}_{0}+1) ({\la}+\kmx(n))} \\
\end{pmatrix}.
\end{gathered}
\end{gather}
\item Apply the obtained gauge transformation $\gbf=\gbf(n,\mv_0,\la)$ 
to the MLR~\er{Mbce},~\er{Ubce} and compute 
the gauge equivalent MLR $\big(\widehat{\mm},\,\widehat{\lt}\big)$ given by~\er{hm2}.
\end{itemize}
We can choose for $a_0=(a^1_0,a^2_0)$ any constant vector 
such that \er{inMmbe} is well defined.
In order to make formula~\er{inMmbe} as simple as possible, we take 
\begin{gather}
\lb{aaa} 
a^1_0=0,\qquad a^2_0=0.
\end{gather}
Then, using~\er{inMmbe} with~\er{aaa}, we see that \er{guum} reads
\begin{gather}
\lb{guu0z}
\gbf(n,\mv_0,\la)=
\sho^{-1}\big(\mm(n,a_0,\mv_{1},\la)^{-1}\big)=
\sho^{-1}\left(\begin{pmatrix}
 \frac{1}{{\la}+\kmx(n)} & -\frac{\mv^{1}_{1}}{{\la}+\kmx(n)} \\
 0 & 1 \\
\end{pmatrix}\right)=
\begin{pmatrix}
 \frac{1}{{\la}+\kmx(n-1)} & -\frac{\mv^{1}_{0}}{{\la}+\kmx(n-1)} \\
 0 & 1 \\
\end{pmatrix}.
\end{gather}
Applying the obtained gauge transformation~\er{guu0z} to the MLR~\er{Mbce},~\er{Ubce} 
and computing the MLR $\big(\widehat{\mm},\,\widehat{\lt}\big)$ given by~\er{hm2}, we get
\begin{gather}
\lb{zchmm}
\widehat{\mm}(n,\mv_{0},\la)=\sho(\gbf)\cdot\mm\cdot\gbf^{-1}=
\begin{pmatrix}
 \frac{({\la}+\kmx(n-1)) (\mv^{1}_{0} \mv^{2}_{0} {\la}+{\la}+\kmx(n)-{\mc} \mv^{1}_{0} \mv^{2}_{0})}{({\la}+\kmx(n)) (\mv^{1}_{0} \mv^{2}_{0}+1)} & \frac{({\la}-{\mc}) \mv^{1}_{0}}{{\la}+\kmx(n)} \\
 ({\la}+\kmx(n-1)) \mv^{2}_{0} & \mv^{1}_{0} \mv^{2}_{0}+1 \\
\end{pmatrix},\\
\lb{zchlt}
\widehat{\lt}(n,\mv_{-1},\mv_{0},\la)=\tdt(\gbf)\cdot\gbf^{-1}+\gbf\cdot\lt\cdot\gbf^{-1}=
\left(\!\begin{smallmatrix}
 2 \mv^{1}_{0} \mv^{2}_{-1}-{\la} & \frac{2 ({\mc}-{\la}) \mv^{1}_{0}}{{\la}+\kmx(n-1)} \\
 -2 ({\la}+\kmx(n-1)) \mv^{2}_{-1} & {\la}-2 \mv^{1}_{0} \mv^{2}_{-1} \\
\end{smallmatrix}\!\right).
\end{gather}
In agreement with Theorem~\ref{thmlpk1}, we see that
the matrix~\er{zchmm} depends only on~$n,\,\mv_{0},\,\la$,
in contrast to the matrix~\er{Mbce} depending on~$n,\,\mv_{0},\,\mv_{1},\,\la$.

Now, using Definition~\ref{dtame} and Remark~\ref{rgnla} with~$\mv$ 
instead of~$u$, we see that the MLR~\er{zchmm},~\er{zchlt} is $\mv$-tame, 
and we can try to simplify it by applying a gauge transformation $G=G(\nt,\la)$ 
depending only on $\nt,\,\la$.
To avoid cumbersome formulas, we take a diagonal gauge transformation 
$G=\left(\begin{smallmatrix}
g^{11}(\nt,\la) & 0\\
0& g^{22}(\nt,\la)
\end{smallmatrix}\right)$, where $g^{11}(\nt,\la)$ and $g^{22}(\nt,\la)$ are nonzero scalar functions.
Then $\sho(G)=\left(\begin{smallmatrix}
g^{11}(\nt+1,\la) & 0\\
0& g^{22}(\nt+1,\la)
\end{smallmatrix}\right)$ and $\tdt(G)=0$.
Applying~$G$ to the MLR~\er{zchmm},~\er{zchlt}, we get the corresponding gauge equivalent MLR
\begin{gather}
\lb{hhm}
\widehat{\widehat{\mm}}=
\sho(G)\cdot\widehat{\mm}(n,\mv_{0},\la)\cdot G^{-1}=
\begin{pmatrix}
 \frac{g^{11}(\nt+1,\la)({\la}+\kmx(n-1)) (\mv^{1}_{0} \mv^{2}_{0} {\la}+{\la}+\kmx(n)-{\mc} \mv^{1}_{0} \mv^{2}_{0})}{({\la}+\kmx(n)) (\mv^{1}_{0} \mv^{2}_{0}+1)g^{11}(\nt,\la)} & 
\frac{g^{11}(\nt+1,\la)({\la}-{\mc}) \mv^{1}_{0}}{({\la}+\kmx(n))g^{22}(\nt,\la)} \\
\frac{g^{22}(\nt+1,\la)({\la}+\kmx(n-1))\mv^{2}_{0}}{g^{11}(\nt,\la)} & 
\frac{g^{22}(\nt+1,\la)\mv^{1}_{0} \mv^{2}_{0}+1}{g^{22}(\nt,\la)} \\
\end{pmatrix},\\
\lb{hhlt}
\widehat{\widehat{\lt}}=\tdt(G)\cdot G^{-1}+G\cdot\widehat{\lt}(n,\mv_{-1},\mv_{0},\la)\cdot G^{-1}=
\left(\!\begin{smallmatrix}
 2 \mv^{1}_{0} \mv^{2}_{-1}-{\la} & 
\frac{2g^{11}(\nt,\la)({\mc}-{\la})\mv^{1}_{0}}{({\la}+\kmx(n-1))g^{22}(\nt,\la)} \\
 \frac{-2g^{22}(\nt,\la)({\la}+\kmx(n-1))\mv^{2}_{-1}}{g^{11}(\nt,\la)} & 
{\la}-2 \mv^{1}_{0} \mv^{2}_{-1} \\
\end{smallmatrix}\!\right).
\end{gather}
We want to choose $G=\left(\begin{smallmatrix}
g^{11}(\nt,\la) & 0\\
0& g^{22}(\nt,\la)
\end{smallmatrix}\right)$ such that the matrices~\er{hhm},~\er{hhlt} become simpler than 
the matrices~\er{zchmm},~\er{zchlt}.
The formulas in the right-hand sides of~\er{hhm},~\er{hhlt} suggest 
to take $g^{11}(\nt,\la)={\la}+\kmx(n-1)$, $\,g^{22}(\nt,\la)=1$.
Then \er{hhm},~\er{hhlt} become
\begin{gather}
\lb{2hhm}
\widehat{\widehat{\mm}}=
\begin{pmatrix}
 \frac{\mv^{1}_{0} \mv^{2}_{0} {\la}+{\la}+\kmx(n)-{\mc} \mv^{1}_{0} \mv^{2}_{0}}{\mv^{1}_{0} \mv^{2}_{0}+1} & ({\la}-{\mc}) \mv^{1}_{0} \\
 \mv^{2}_{0} & \mv^{1}_{0} \mv^{2}_{0}+1 \\
\end{pmatrix},\qquad\quad
\widehat{\widehat{\lt}}=
\begin{pmatrix}
 2 \mv^{1}_{0} \mv^{2}_{-1}-{\la} & 2 ({\mc}-{\la}) \mv^{1}_{0} \\
 -2 \mv^{2}_{-1} & {\la}-2 \mv^{1}_{0} \mv^{2}_{-1} \\
\end{pmatrix}.
\end{gather}

Fix a constant $\ap\in\mathbb{C}$.
Below we use the matrices~\er{2hhm} with~$\la$ replaced by~$\ap$.
According to Definition~\ref{dmlpgt}, since the matrices~\er{2hhm} 
form a MLR for equation~\er{dtmv}, one can consider the auxiliary system
\begin{gather}
\lb{tsyspsit}
\sho(\tilde{\Psi})=\begin{pmatrix}
 \frac{\mv^{1}_{0} \mv^{2}_{0} {\ap}+{\ap}+\kmx(n)-{\mc} \mv^{1}_{0} \mv^{2}_{0}}{\mv^{1}_{0} \mv^{2}_{0}+1} 
& ({\ap}-{\mc}) \mv^{1}_{0} \\
 \mv^{2}_{0} & \mv^{1}_{0} \mv^{2}_{0}+1 \\
\end{pmatrix}\cdot\tilde{\Psi},\qquad\qquad
\pd_t(\tilde{\Psi})=\begin{pmatrix}
 2 \mv^{1}_{0} \mv^{2}_{-1}-{\ap} & 2 ({\mc}-{\ap}) \mv^{1}_{0} \\
 -2 \mv^{2}_{-1} & {\ap}-2 \mv^{1}_{0} \mv^{2}_{-1} \\
\end{pmatrix}\cdot\tilde{\Psi},
\end{gather}
which is compatible modulo equation~\eqref{dtmv}. 
Here $\tilde{\Psi}=\tilde{\Psi}(n,t)$ is an invertible $2\times 2$ matrix 
with elements $\tilde{\uppsi}^{ij}=\tilde{\uppsi}^{ij}(n,t)$, $\,i,j=1,2$.

Substituting $\tilde{\Psi}=\begin{pmatrix}
\tilde{\uppsi}^{11}&\tilde{\uppsi}^{12}\\
\tilde{\uppsi}^{21}&\tilde{\uppsi}^{22}
\end{pmatrix}$ in~\er{tsyspsit} and considering the first columns 
in the matrix equations~\er{tsyspsit}, we obtain
\begin{gather}
\lb{tsp11}
\sho(\tilde{\uppsi}^{11})=
\frac{(\mv^{1}_{0} \mv^{2}_{0} {\ap}+{\ap}+\kmx(n)
-{\mc} \mv^{1}_{0} \mv^{2}_{0})}{\mv^{1}_{0} \mv^{2}_{0}+1}\tilde{\uppsi}^{11}
+({\ap}-{\mc})\mv^{1}_{0}\tilde{\uppsi}^{21},\qquad\quad
\sho(\tilde{\uppsi}^{21})=\mv^{2}_{0}\tilde{\uppsi}^{11}+
(\mv^{1}_{0}\mv^{2}_{0}+1)\tilde{\uppsi}^{21},\\
\lb{tdp11}
\pd_t(\tilde{\uppsi}^{11})=(2\mv^{1}_{0}\mv^{2}_{-1}-{\ap})\tilde{\uppsi}^{11}
+2({\mc}-{\ap})\mv^{1}_{0}\tilde{\uppsi}^{21},\qquad\quad
\pd_t(\tilde{\uppsi}^{21})=-2 \mv^{2}_{-1}\tilde{\uppsi}^{11}
+({\ap}-2 \mv^{1}_{0} \mv^{2}_{-1})\tilde{\uppsi}^{21}.
\end{gather}
We set $\tvv={\tilde{\uppsi}^{11}}/{\tilde{\uppsi}^{21}}$.
Similarly to the computations~\er{swp11},~\er{dwp11} based on~\er{sp11},~\er{dp11},
now we compute $\sho(\tvv)=\sho\big({\tilde{\uppsi}^{11}}/{\tilde{\uppsi}^{21}}\big)$ and 
$\pd_t(\tvv)=\pd_t\big({\tilde{\uppsi}^{11}}/{\tilde{\uppsi}^{21}}\big)$, 
using~\er{tsp11},~\er{tdp11}, and get
\begin{gather}
\label{tswuu}
\sho(\tvv)=\frac{(\mv^{1}_{0} \mv^{2}_{0} {\ap}+{\ap}+\kmx(n)
-{\mc} \mv^{1}_{0} \mv^{2}_{0})\tvv+({\ap}-{\mc})\mv^{1}_{0}(\mv^{1}_{0} \mv^{2}_{0}+1)}%
{(\mv^{1}_{0} \mv^{2}_{0}+1)(\mv^{2}_{0}\tvv+\mv^{1}_{0}\mv^{2}_{0}+1)},\\
\label{tdwuu}
\pd_t(\tvv)=2\mv^{2}_{-1}(\tvv)^2+(4 \mv^{1}_{0} \mv^{2}_{-1}-2{\ap})\tvv+2 ({\mc}-{\ap}) \mv^{1}_{0}.
\end{gather}
The described way of deriving equations~\er{tswuu},~\er{tdwuu} 
from system~\er{tsyspsit}, which is compatible modulo equation~\eqref{dtmv},
implies that system~\er{tswuu},~\er{tdwuu} is compatible modulo~\eqref{dtmv} as well. 

Now we set
\begin{gather}
\label{tvwu}
\bv^1:=\tvv,\qquad\qquad\bv^2:=\mv^{2}.
\end{gather}
One has $\sho(\tvv)=\sho(\bv^1)=\bv^1_1$, $\,\mv^i_0=\mv^i$, $\,\bv^i_0=\mv^i$, $\,i=1,2$.
Substituting $\sho(\tvv)=\bv^1_1$, $\,\tvv=\bv^1_0$, $\,\mv^{1}_{0}=\mv^{1}$, $\,\mv^{2}_{0}=\bv^2_0$ 
in~\er{tswuu}, we get
\begin{gather}
\label{mvv11}
\bv^1_1=\frac{(\mv^{1}\bv^2_0{\ap}+{\ap}+\kmx(n)
-{\mc}\mv^{1}\bv^2_0)\bv^1_0+({\ap}-{\mc})\mv^{1}(\mv^{1}\bv^2_0+1)}%
{(\mv^{1}\bv^2_0+1)(\bv^2_0\bv^1_0+\mv^{1}\bv^2_0+1)}.
\end{gather}
Equation~\er{mvv11} can be resolved with respect to~$\mv^{1}$ as follows
\begin{gather}
\label{treseq}
\mv^{1}=\frac{\mX+\sqrt{(\mX)^2+2%
\big({\ap}\bv^{1}_{0}+\kmx(n)\bv^{1}_{0}-\bv^{1}_{1}\bv^{2}_{0}\bv^{1}_{0}-\bv^{1}_{1}\big)\mY}}%
{\mY},
\end{gather}
where $\mX$, $\mY$ are given by
\begin{gather}
\label{xy}
\mX:=\bv^{2}_{0}(\bv^{1}_{0}({\ap}-\bv^{1}_{1}\bv^{2}_{0}-{\mc})-2\bv^{1}_{1})-{\mc}+{\ap},\qquad\quad
\mY:=2\bv^{2}_{0}(\bv^{1}_{1}\bv^{2}_{0}+{\mc}-{\ap}).
\end{gather}
Now consider equation~\er{treseq} along with $\mv^{2}=\bv^2_0$, which gives
\begin{gather}
\label{tmtkmx}
\left\{
\begin{aligned}
\mv^1&=\frac{\mX+\sqrt{(\mX)^2+2%
\big({\ap}\bv^{1}_{0}+\kmx(n)\bv^{1}_{0}-\bv^{1}_{1}\bv^{2}_{0}\bv^{1}_{0}-\bv^{1}_{1}\big)\mY}}%
{\mY},\\ 
\mv^2&=\bv^{2}_{0}.
\end{aligned}\right.
\end{gather}
For each $\ell\in\zz$, applying the operator~$\sho^\ell$ to~\er{tmtkmx}, we get
\begin{gather}
\label{tuellmt}
\left\{
\begin{aligned}
\mv^1_\ell&=\frac{\sho^\ell(\mX)+\sqrt{(\sho^\ell(\mX))^2+2\big({\ap}\bv^{1}_{\ell}+%
\kmx(n+\ell)\bv^{1}_{\ell}-\bv^{1}_{\ell+1}\bv^{2}_{\ell}\bv^{1}_{\ell}-\bv^{1}_{\ell+1}\big)\sho^\ell(\mY)}}%
{\sho^\ell(\mY)},\\ 
\mv^2_\ell&=\bv^2_\ell,\qquad\qquad\qquad\qquad\ell\in\zz.
\end{aligned}\right.
\end{gather}
Substituting~$\tvv=\bv^1$ and~\er{tuellmt} in~\er{tdwuu}, we derive
\begin{gather}
\label{tdtmv1}
\pd_t(\bv^1)=2\bv^{2}_{-1}(\bv^1_0)^2-2{\ap}\bv^1_0+
\frac{\big(4\bv^{2}_{-1}\bv^1_0+2({\mc}-{\ap})\big)\left(\mX\!+\!\sqrt{(\mX)^2\!+\!2%
\big({\ap}\bv^{1}_{0}\!+\!\kmx(n)\bv^{1}_{0}\!-\!\bv^{1}_{1}\bv^{2}_{0}\bv^{1}_{0}\!-\!\bv^{1}_{1}\big)\mY}\right)}%
{\mY}.
\end{gather}
Now we consider~\er{tdtmv1} along with the equation obtained 
by means of substituting~\er{tuellmt} in the second component of~\er{dtmv}.
This gives the $2$-component non-autonomous equation
\begin{gather}
\label{tdtmv}
\left\{
\begin{aligned}
\pd_t(\bv^1)&=2\bv^{2}_{-1}(\bv^1_0)^2-2{\ap}\bv^1_0+
\frac{\big(4\bv^{2}_{-1}\bv^1_0+2({\mc}-{\ap})\big)\left(\mX\!+\!\sqrt{(\mX)^2\!+\!2%
\big({\ap}\bv^{1}_{0}\!+\!\kmx(n)\bv^{1}_{0}\!-\!\bv^{1}_{1}\bv^{2}_{0}\bv^{1}_{0}\!-\!\bv^{1}_{1}\big)\mY}\right)}%
{\mY},\\ 
\pd_t(\bv^2)&=
\mH\big(\nt,\bv^{1}_{0},\bv^{1}_{1},\bv^{1}_{2},\bv^{2}_{-1},\bv^{2}_{0},\bv^{2}_{1},\mc,\ap\big),
\end{aligned}\right.
\end{gather}
where the function 
$\mH=\mH\big(\nt,\bv^{1}_{0},\bv^{1}_{1},\bv^{1}_{2},\bv^{2}_{-1},\bv^{2}_{0},\bv^{2}_{1},\mc,\ap\big)$ 
is obtained by substituting~\er{tuellmt} in the right-hand side of the second component of~\er{dtmv}.
(We do not present the explicit formula for~$\mH$, since it is rather cumbersome.)

As discussed above, 
system~\er{tswuu},~\er{tdwuu} is compatible modulo equation~\eqref{dtmv}.
Since equations \er{tmtkmx}, \er{tdtmv} 
are derived from~\er{tswuu},~\er{tdwuu},~\er{dtmv} (by means of introducing~\er{tvwu}), 
system~\er{tmtkmx},~\er{tdtmv} is compatible modulo equation~\eqref{dtmv} as well.
This implies that \er{tmtkmx} is a MT from~\er{tdtmv} to~\er{dtmv}.
(This fact can also be verified by a direct calculation.)

\begin{remark}
\lb{rlax1}
Recall that \er{2hhm} is a MLR for equation~\er{dtmv},
and formulas~\er{tuellmt} are derived by applying the operator~$\sho^\ell$ to~\er{tmtkmx}.
Using Lemma~\ref{luv} for the MLR~\er{2hhm} and 
the MT~\er{tmtkmx} from~\er{tdtmv} to~\er{dtmv},
we can substitute~\er{tuellmt} in~\er{2hhm}, and this gives a MLR for equation~\er{tdtmv}.
\end{remark}

\subsection{Other examples}
\lb{sotex}

The following $2$-component equation 
was derived in~\cite[Section~5.1.1]{MMBW2013} 
in a study of invariant time evolutions of polygons in the centro-affine plane
\begin{gather}
\label{MMBW}
\left\{
\begin{aligned}
\pd_t(\pe^1)&=\pe^{1}_{0} \pe^{2}_{0},\\ 
\pd_t(\pe^2)&=\frac{\pe^{1}_{-1}}{\pe^{1}_{0}}-\frac{\pe^{1}_{0}}{\pe^{1}_{1}}.
\end{aligned}\right.
\end{gather}
It is shown in~\cite[Section~5.1.1]{MMBW2013} that one has the MT 
\begin{gather}
\label{mtmbw}
\left\{
\begin{aligned}
\td^1&=\pe^{1}_{0}/\pe^{1}_{1},\\ 
\td^2&=\pe^{2}_{0}
\end{aligned}\right.
\end{gather}
from~\er{MMBW} to the equation
\begin{gather}
\label{todafm}
\left\{
\begin{aligned}
\pd_t(\td^1)&=\td^{1}_{0} (\td^{2}_{0}-\td^{2}_{1}),\\ 
\pd_t(\td^2)&=\td^{1}_{-1}-\td^{1}_{0},
\end{aligned}\right.
\end{gather}
which is the Toda lattice in the Flaschka--Manakov coordinates~\cite{FlaschPhR74,Manak75}.

It is known (see, e.g.,~\cite{kmw}) that the following matrices form a MLR for~\eqref{todafm}
\begin{gather}
\lb{mlrtd}
\breve{\mm}=\begin{pmatrix}
 {\la}-\td^{2}_{1} & \td^{1}_{0} \\
 -1 & 0 \\
\end{pmatrix},
\qquad\quad
\breve{\lt}=\begin{pmatrix}
 0 & -\td^{1}_{0} \\
 1 & {\la}-\td^{2}_{0} \\
\end{pmatrix}.
\end{gather}
Applying Lemma~\ref{luv} to the MLR~\er{mlrtd} and 
to the MT~\er{mtmbw} from~\er{MMBW} to~\er{todafm},
we substitute~\er{mtmbw} in~\er{mlrtd} and obtain the following MLR for equation~\er{MMBW}
\begin{gather}
\lb{tMbce}
\mm=\begin{pmatrix}
 {\la}-\pe^{2}_{1} & \pe^{1}_{0}/\pe^{1}_{1}\\
 -1 & 0 \\
\end{pmatrix},\\
\lb{tUbce}
\lt=\begin{pmatrix}
 0 & -\pe^{1}_{0}/\pe^{1}_{1}\\
 1 & {\la}-\pe^{2}_{0} \\
\end{pmatrix}.
\end{gather}
In order to simplify the MLR~\er{tMbce},~\er{tUbce} by means of gauge transformations,
we consider Theorem~\ref{thmlpk1} in the case $\ndu=2$, $\,\sm=2$, $\,\ip=1$.
The matrix~\er{tMbce} obeys condition~\er{pdm} with~$\ip=1$.
Therefore, one can apply Theorem~\ref{thmlpk1} to the MLR~\er{tMbce}, \er{tUbce}.
To this end, we need to do the following.
\begin{itemize}
	\item Choose a constant vector $a_0=(a^1_0,a^2_0)\in\fik^2$.
	\item Substitute $\pe_0=a_0$ in~$\mm$
given by~\er{tMbce} (i.e., substitute $\pe^i_0=a^i_0$ for $i=1,2$ in~\er{tMbce}) 
and compute the corresponding gauge transformation~\er{guum}, 
which gives $\gbf=\begin{pmatrix}
 0 & -1 \\
 \pe^{1}_{0}/a^1_0 & \pe^{1}_{0} ({\la}-\pe^{2}_{0})/a^1_0
\end{pmatrix}$.
Here it is convenient to take $a^1_0=1$, and then $\gbf$ takes the form 
$\gbf=\begin{pmatrix}
 0 & -1 \\
 \pe^{1}_{0} & \pe^{1}_{0} ({\la}-\pe^{2}_{0})
\end{pmatrix}$.
\item Now we apply the obtained gauge transformation $\gbf$ 
to the MLR~\er{tMbce},~\er{tUbce} and compute 
the gauge equivalent MLR $\big(\widehat{\mm},\,\widehat{\lt}\big)$ given by~\er{hm2}, which gives
\begin{gather}
\lb{hmmtd}
\widehat{\mm}=\sho(\gbf)\cdot\mm\cdot\gbf^{-1}=
\begin{pmatrix}
 {\la}-\pe^{2}_{0} & 1/\pe^{1}_{0}\\
 -\pe^{1}_{0} & 0 \\
\end{pmatrix},\qquad
\widehat{\lt}=\tdt(\gbf)\cdot\gbf^{-1}+\gbf\cdot\lt\cdot\gbf^{-1}=
\begin{pmatrix}
 0 & -1/\pe^{1}_{0}\\
 \pe^{1}_{-1} & {\la} \\
\end{pmatrix}.
\end{gather}
\end{itemize}

Fix a constant $\mc\in\mathbb{C}$.
Below we use the matrices~\er{hmmtd} with~$\la$ replaced by~$\mc$.
According to Definition~\ref{dmlpgt}, since \er{hmmtd} 
is a MLR for equation~\er{MMBW}, we can consider the auxiliary system
\begin{gather}
\lb{astd}
\sho(\Psi)=\begin{pmatrix}
 {\mc}-\pe^{2}_{0} & 1/\pe^{1}_{0}\\
 -\pe^{1}_{0} & 0 \\
\end{pmatrix}\cdot\Psi,\qquad\quad
\pd_t(\Psi)=\begin{pmatrix}
 0 & -1/\pe^{1}_{0}\\
 \pe^{1}_{-1} & {\mc} \\
\end{pmatrix}\cdot\Psi,
\end{gather}
which is compatible modulo equation~\eqref{MMBW}. 
Here $\Psi=\Psi(n,t)$ is an invertible $2\times 2$ matrix
with elements $\uppsi^{ij}=\uppsi^{ij}(n,t)$, $\,i,j=1,2$.

Substituting $\Psi=\begin{pmatrix}
\uppsi^{11}&\uppsi^{12}\\
\uppsi^{21}&\uppsi^{22}
\end{pmatrix}$ in~\er{astd} and considering the first columns 
in the matrix equations~\er{astd}, one gets
\begin{gather}
\lb{usp11}
\sho(\uppsi^{11})=({\mc}-\pe^{2}_{0})\uppsi^{11}
+\frac{1}{\pe^{1}_{0}}\uppsi^{21},\qquad\quad
\sho(\uppsi^{21})=-\pe^{1}_{0}\uppsi^{11},\\
\lb{udp11}
\pd_t(\uppsi^{11})=-\frac{1}{\pe^{1}_{0}}\uppsi^{21},\qquad\quad
\pd_t(\uppsi^{21})=\pe^{1}_{-1}\uppsi^{11}+{\mc}\uppsi^{21}.
\end{gather}
Since the matrix system~\er{astd} is compatible modulo equation~\eqref{MMBW}, 
the corresponding system~\er{usp11},~\er{udp11} 
obtained from the first columns in~\er{astd} is compatible modulo~\eqref{MMBW} as well.
We set
\begin{gather}
\label{mvup}
\mv^1:=\uppsi^{11},\qquad\qquad\mv^2:=\uppsi^{21}.
\end{gather}
For $i=1,2$ one has $\sho(\uppsi^{i1})=\sho(\mv^i)=\mv^i_1$, 
$\,\pe^i_0=\pe^i$, $\,\mv^i_0=\mv^i$.
Substituting $\sho(\uppsi^{i1})=\mv^i_1$, $\,\uppsi^{i1}=\mv^i_0$, $\,\pe^i_0=\pe^i$ 
in~\er{usp11}, we get
\begin{gather}
\lb{mvi1}
\mv^1_1=({\mc}-\pe^{2})\mv^1_0+\frac{1}{\pe^{1}}\mv^2_0,
\qquad\quad\mv^2_1=-\pe^{1}\mv^1_0.
\end{gather}
Equations~\er{mvi1} can be resolved with respect to~$\pe^1$, $\pe^2$ as follows
\begin{gather}
\label{mtmtd}
\left\{
\begin{aligned}
\pe^1&=-\mv^{2}_{1}/\mv^{1}_{0},\\ 
\pe^2&=({\mc}\mv^{1}_{0}\mv^{2}_{1}-\mv^{1}_{0}\mv^{2}_{0}-\mv^{1}_{1}%
\mv^{2}_{1})/(\mv^{1}_{0} \mv^{2}_{1}).
\end{aligned}\right.
\end{gather}
For each $\ell\in\zz$, applying the operator~$\sho^\ell$ to~\er{mtmtd}, we get
\begin{gather}
\label{pe12ell}
\left\{
\begin{aligned}
\pe^1_\ell&=-\mv^{2}_{\ell+1}/\mv^{1}_{\ell},\\ 
\pe^2_\ell&=({\mc}\mv^{1}_{\ell}\mv^{2}_{\ell+1}-\mv^{1}_{\ell}\mv^{2}_{\ell}-\mv^{1}_{\ell+1}%
\mv^{2}_{\ell+1})/(\mv^{1}_{\ell} \mv^{2}_{\ell+1}),\qquad\qquad\qquad\qquad\ell\in\zz.
\end{aligned}\right.
\end{gather}
Substituting $\uppsi^{11}=\mv^1$, $\,\uppsi^{21}=\mv^2$ 
and~\er{pe12ell} in~\er{udp11}, we obtain $2$-component equation
\begin{gather}
\label{ceqnew}
\left\{
\begin{aligned}
\pd_t(\mv^1)&=\mv^{1}_{0} \mv^{2}_{0}/\mv^{2}_{1},\\ 
\pd_t(\mv^2)&=\mv^{2}_{0} ({\mc} \mv^{1}_{-1}-\mv^{1}_{0})/\mv^{1}_{-1}.
\end{aligned}\right.
\end{gather}
As discussed above, 
system~\er{usp11},~\er{udp11} is compatible modulo equation~\eqref{MMBW}.
Since equations~\er{mtmtd},~\er{ceqnew} 
are derived from~\er{usp11},~\er{udp11} (by means of introducing~\er{mvup}), 
system~\er{mtmtd},~\er{ceqnew} is compatible modulo equation~\eqref{MMBW} as well.
This implies that \er{mtmtd} is a MT from~\er{ceqnew} to~\er{MMBW},
which can also be verified by a straightforward computation.

\begin{remark}
\lb{rlax2}
Recall that \er{hmmtd} is a MLR for equation~\er{MMBW},
and formulas~\er{pe12ell} are obtained by applying the operator~$\sho^\ell$ to~\er{mtmtd}.
Using Lemma~\ref{luv} for the MLR~\er{hmmtd} and 
the MT~\er{mtmtd} from~\er{ceqnew} to~\er{MMBW},
we can substitute~\er{pe12ell} in~\er{hmmtd}, and this gives a MLR for equation~\er{ceqnew}.
\end{remark}

\section*{Acknowledgments}

The author would like to thank S.~Konstantinou-Rizos and A.V.~Mikhailov for useful discussions.

This work was supported by the Russian Science Foundation (grant No. 25-21-00454, 
\url{https://rscf.ru/en/project/25-21-00454/} ).

\end{document}